\documentclass[letterpaper, 11pt]{amsart}
\usepackage{macros}
\usepackage{fullpage}
\usepackage[most]{tcolorbox}
\usepackage{setspace}

\title{Spatial Mixing and Deterministic Approximate Counting of Multi-spin Systems beyond Bounded Degree Graphs}
\usepackage[foot]{amsaddr}
\newboolean{doubleblind}
\setboolean{doubleblind}{False}
\ifdoubleblind
    \author{Author(s)}
\else
    \author{Zhidan Li, Kuan Yang}
    \address[Zhidan Li]{School of Computer Science, Shanghai Jiao Tong University, Shanghai, China. \textnormal{Email: \url{yueEnTeRnAL@sjtu.edu.cn}}.}
    \address[Kuan Yang]{John Hopcroft Center for Computer Science, Shanghai Jiao Tong University, Shanghai, China. \textnormal{Email: \url{kuan.yang@sjtu.edu.cn}}.}
\fi

\usepackage{notations}
\numberwithin{equation}{section}

\begin{document}

\begin{abstract}
    We develop a framework for deterministic approximate counting of multi-spin systems beyond bounded-degree graphs. 
    The algorithm recursively constructs rational polytopes containing the true marginal vectors and uses linear-fractional programming to obtain certified bounds on marginal ratios. 
    For positive interactions on graphs of polynomial connective constant $D$, we establish strong spatial mixing whenever $Dc<1$, where $c$ bounds the Birkhoff contraction coefficients of the interactions. Under the same condition, we also give a deterministic fully polynomial-time approximation scheme (\textbf{FPTAS}) for the partition function.
    
    We further extend the framework to proper colorings of sparse Erd\H{o}s-R\'{e}nyi random graphs using recursion on permissive blocks. For every fixed $\eta\in(0,1)$, sufficiently large fixed $d$, and fixed integer $q\ge(2+\eta)d$, we obtain an \textbf{FPTAS} for counting proper $q$-colorings of $G\sim\mathcal G(n,d/n)$ with high probability over $G$.
    This improves the leading constant $3$ in the earlier counting guarantee of Yin and Zhang (APPROX/RANDOM, 2016) to $2$, and asymptotically matches the spatial mixing regime established by Yin (ICALP, 2014).
\end{abstract}

\maketitle

\section{Introduction} \label{sec:introduction}
The multi-spin system is a statistical physics model reflecting the global behavior of local interactions.
Their partition functions encode many counting problems in combinatorics, as well as key thermodynamic quantities in statistical mechanics. Exact evaluation of partition functions is $\#\mathbf P$-hard for broad classes of interactions \cite{CC19,CCL13}, making efficient approximation algorithms a central object of study in theoretical computer science.

An instance of $q$-spin systems is a tuple $\Phi = (G = (V, E), [q], \boldsymbol A = (A_e)_{e \in E}, \boldsymbol\lambda = (\lambda_v)_{v \in V})$ where $G = (V, E)$ is the underlying graph, the edge interaction $A_e \in \mathbb R_{\ge 0}^{q \times q}$ is a non-zero symmetric matrix for every $e \in E$, and $\lambda_v \in \mathbb R_{> 0}^q$ is the external field associated with a vertex $v \in V$.
The partition function of $\Phi$ is defined as
\begin{align*}
    \Z(\Phi) = \Z_G(\boldsymbol\lambda, \boldsymbol A) \defeq \sum_{\sigma : V \to [q]} w(\sigma), \quad \text{where } \  w(\sigma) = \prod_{v \in V} \lambda_v(\sigma(v)) \prod_{e = (u, v) \in E} A_e(\sigma(u), \sigma(v))
\end{align*}
gives the weight of configuration $\sigma$.
The corresponding Gibbs distribution $\mu_\Phi$ is then defined as
\begin{align} \label{eq:spin-Gibbs-distribution}
    \forall \sigma \in [q]^V, \quad \mu_\Phi(\sigma) = \frac{w(\sigma)}{\Z_G(\boldsymbol\lambda, \boldsymbol A)}.
\end{align}

Correlation decay provides a powerful route to deterministically approximating partition functions.
For a self-reducible problem, the partition function can be recovered from a sequence of marginal probabilities. A correlation-decay algorithm estimates these marginals by truncating a suitable recursion and controlling the effect of the resulting boundary error.
This approach has led to fully polynomial-time deterministic approximation schemes (\textbf{FPTAS}es) for the hard-core model \cite{Weitz06}, matchings \cite{BGKNT07}, colorings and multi-spin systems \cite{GK12,LY13}, and counting problems in the Holant framework \cite{YZ13,LWZ14}.
An underlying probabilistic property is \emph{strong spatial mixing} (SSM), which indicates that the influence of a boundary condition decays exponentially with the distance to the vertices where that condition changes.



Unfortunately, when the number of spins $q$ is greater than $2$, turning spatial mixing into an efficient algorithm often presents a challenge.
To obtain an \textbf{FPTAS}, a recursive marginal estimator must control errors for the approximate messages that it actually produces.
However, prior analyses over independent coordinate ranges may lose constraints shared by the coordinates of a probability vector, including normalization and model-specific marginal bounds.
The resulting contraction condition can be stronger than the condition needed to compare true marginals.
In addition, if the underlying graphs are not bounded-degree, vertices of large degree can cause both local error amplification and a large number of recursive calls.
The connective-constant approach controls the growth of self-avoiding walks and has allowed correlation-decay methods to handle certain models beyond bounded-degree graphs \cite{SSY13,SSSY17}.
Extending this approach to multi-spin systems requires simultaneous control of the admissible messages and the cost of the recursion.
For example, counting proper colorings of sparse random graphs illustrates both difficulties.
For $G\sim\mathcal G(n,d/n)$, Yin \cite{Yin14} established, with high probability, strong spatial mixing at an arbitrarily fixed vertex when $q\ge\alpha d+\beta$ for every fixed $\alpha>2$ and a sufficiently large constant $\beta$.
The deterministic counting method of Yin and Zhang \cite{YZ16} applies when $q>3d+4$.
Closing this gap requires a recursion that preserves more of the structure of marginal distributions while accommodating vertices of unbounded degree.
For approximating the free energy of random graph colorings, Efthymiou gives a different deterministic method that applies as long as $q\ge(1+\delta)d$ for any fixed $\delta>0$ \cite{Efthymiou20}.
That guarantee concerns the free energy rather than a relative approximation to the partition function itself: a $(1+\eps)$ approximation to $Z$ would require additive error at most $\ln(1+\eps)$ in $\ln Z$.

In this work, we study strong spatial mixing and deterministic approximate counting for multi-spin systems on graphs with potentially unbounded degree, but with bounded \emph{connective constant}.
The connective constant is a quantity associated with the number of self-avoiding walks.
Bounded connective constant typically allows for locally dense parts in the underlying graphs, while from the global perspective, such graphs remain similar to those with bounded degree.

We first extend the connective-constant results in~\cite{SSY13,SSSY17} to positive multi-spin systems and prove spatial mixing with an explicit polynomial prefactor. Then we develop a new framework to design $\mathbf{FPTAS}$es based on correlation decay.
Our main algorithm recursively constructs \emph{rational polytopes containing the true marginal vector}. The polytopes retain constraints among marginal coordinates, and linear-fractional optimization produces certified bounds for each factor in the exact marginal-ratio recursion. 
We additionally apply our algorithm to approximate the number of proper $q$-colorings of sparse random graphs.
We control the total cost of processing permissive blocks and improve the leading constant in the sufficient condition for deterministic approximate counting.


\subsection{Main results}
Our first result treats multi-spin systems with positive interactions on graphs of bounded connective constant.
Writing $D$ for a bound on the growth of self-avoiding walks and $c$ for a uniform bound on the Birkhoff contraction coefficients of the interactions, we establish spatial mixing with a polynomial prefactor whenever $Dc<1$.
We then give a fully polynomial-time approximation scheme (\textbf{FPTAS}) for the partition function under the same condition.
The external fields may vary arbitrarily across vertices, and the graph need not have bounded maximum degree.

Our second result concerns proper colorings of sparse Erd\H{o}s--R\'enyi random graphs.
For every fixed $\eta\in(0,1)$, sufficiently large fixed $d$, and fixed integer $q\ge(2+\eta)d$, we give a deterministic FPTAS that succeeds with high probability over $G\sim\mathcal G(n,d/n)$.
This improves the leading constant in the sufficient counting condition of \cite{YZ16} from $3$ to $2$ for sufficiently large $d$.
It agrees only to leading order with the sufficient condition in \cite{Yin14}, whose spatial-mixing conclusion concerns a vertex fixed before the random graph is drawn.
Here the high-probability guarantee is over the input graph; the counting algorithm itself is deterministic.

\subsubsection*{Computation model}
Before formally stating our main result, we clarify the computation model here.
In this work, we employ the classical \emph{deterministic Turing machine}, which can be described as follows.    Rational numbers are encoded by binary bits, and we allow all common bit operations, \EG, arithmetic operations and shifting operations.
For an input $\Phi$, the length of $\Phi$ is the total number of bits to encode $\Phi$ in the model.
We use $\abs\Phi$ to denote its length.
For a single element $x$ in the model, we use $\bit(x)$ to denote the number of bits to encode $x$.
The time complexity of an algorithm is the bit complexity of it, \IE, the number of bit operations used by it.
All algorithmic input instances have rational interaction matrices and external fields, encoded in binary.
For $Z>0$ and $\eps\in(0,1)$, an $\eps$-approximation is a number $\widehat Z$ satisfying $(1-\eps)Z\le\widehat Z\le(1+\eps)Z$.
The algorithms below satisfy the stronger one-sided guarantee $Z\le\widehat Z\le(1+\eps)Z$.

\vspace{0.5em}

Now we formally introduce our results.
Our first result is on the \emph{spatial mixing} of multi-spin systems with all positive interactions.
For two probability distributions $p, p'$ on $\Omega$, the total variation distance between $p$ and $p'$ is defined as
\begin{align*}
    \TV{p}{p'} = \frac12 \sum_{c \in \Omega} \abs{p(c) - p'(c)}.
\end{align*}
\begin{definition}[Strong spatial mixing (with a polynomial prefactor)] \label{def:spatial-mixing}
    A family of Gibbs distributions on finite graphs exhibits \emph{strong spatial mixing (with a polynomial prefactor)} of rate $\rho\in(0,1)$ if there is a polynomial $p$ such that, for every distribution $\mu$ in the family on a graph $G=(V,E)$, every pair of feasible partial assignments $\sigma,\tau$ on $\Lambda\subseteq V$, and every $v\in V\setminus\Lambda$,
    \begin{align*}
        \TV{\mu(v \gets \cdot \;|\; \Lambda \gets \sigma)}{\mu(v \gets \cdot \;|\; \Lambda \gets \tau)} \le p(\abs V) \rho^{\dist_G(v, \Delta(\sigma, \tau))}.
    \end{align*}
    Here feasibility means that both conditioning events have positive probability, and
    \[
        \Delta(\sigma,\tau)\defeq\set{u\in\Lambda\mid\sigma(u)\ne\tau(u)}.
    \]
    The notation $\dist_G$ denotes graph distance, with distance $\infty$ to the empty set or an unreachable set and $\rho^\infty=0$.
    The polynomial $p$ is uniform over the family: it may depend on the fixed model and graph-family parameters, but not on the individual instance or boundary conditions.
    We say that the family exhibits \emph{uniform strong spatial mixing} if the same bound holds with a constant prefactor $p(n)=C$ independent of the graph size, the instance, and the boundary conditions.
\end{definition}

For a symmetric matrix $A$ with positive entries, define its \emph{Birkhoff projective diameter} as
\begin{align} \label{eq:Birknoff-projective-diameter}
    \Delta_B(A) \defeq \ln\max_{i, j, k, \ell \in [q]} \frac{A_{ik} A_{j\ell}}{A_{i \ell} A_{j k}}
\end{align}
and the \emph{Birkhoff contraction coefficient} as
\begin{align} \label{eq:Birknoff-contraction-coefficient}
    c_B(A) \defeq \tanh\ab(\frac{\Delta_B(A)}{4}).
\end{align}
The coefficient $c_B(A)$ measures the contraction of positive vectors under multiplication by $A$ in the Hilbert metric.

We will show that the Gibbs distribution of a multi-spin system exhibits the following spatial mixing property related to the \emph{connective constant}\footnote{
For simplicity in the algorithm, we use the finite-graph self-avoiding-walk bound with a polynomial prefactor, as in \cite[Definition~5]{YZ16}.
The difference between it and the standard one does not affect \Cref{thm:multi-spin-system}.
}. For convenience, we use the following characterization of connective constants.
A self-avoiding walk of $G = (V, E)$ of length $\ell$ is a simple path $v_0, \ldots, v_\ell$ in $G$.
Denote by $\SAW_G(v, \ell)$ the set of self-avoiding walks of length $\ell$ starting from $v$, and define
$$
    S_\ell(G) \defeq \max_{v \in V} \;\abs{\SAW_G(v, \ell)}
$$
with convention $S_0(G) = 1$.
We say $\boldsymbol{\mathcal G} = (\mathcal G_n)_{n \ge 1}$ where $\mathcal G_n$ is a family of $n$-vertex graphs has a \emph{polynomial connective constant} $D \ge 1$ if there are constants $K \ge 1$ and $\alpha \ge 0$ such that
\begin{align} \label{eq:polynomial-connective-constant}
    \forall n \ge 1, \ell \ge 0, G \in \mathcal G_n, \quad S_\ell(G) \le K n^\alpha D^\ell.
\end{align}
\begin{theorem}[Spatial mixing] \label{thm:multi-spin-system-spatial-mixing}
    Fix an integer $q \ge 3$, and constants $K \ge 1, \alpha \ge 0, c > 0 $ and $D \ge 1$ such that $D  c < 1$.
    Let $\boldsymbol{\mathcal G} = (\mathcal G_n)_{n \ge 1}$ satisfy \eqref{eq:polynomial-connective-constant}.
    The family of Gibbs distributions $\mu_\Phi$ over instances $\Phi = (G, [q], \boldsymbol A, \boldsymbol\lambda)$ with positive interactions, $G \in \mathcal G_n$ and $c_B(A_e) \le c$ for every $e \in E$, exhibits strong spatial mixing at rate $D c$, with a polynomial prefactor and a coefficient depending only on $K,\alpha,D,c$. Specifically, for feasible partial assignments $\sigma,\tau$ on $\Lambda\subseteq V$ and $v\in V\setminus\Lambda$, 
    \begin{align} \label{eq:spatial-mixing-explicit}
        \TV{\mu_\Phi(v\gets\cdot\mid\Lambda\gets\sigma)}        {\mu_\Phi(v\gets\cdot\mid\Lambda\gets\tau)}
        \le C_* n^\alpha\rho^r,
        \qquad
        \text{ if } r < \infty,
    \end{align}
    where $r=\dist_G(v,\Delta(\sigma,\tau)), \rho=Dc$, and $C_*$ is a constant depending only on $K, \alpha, D, c$.
    If $r=\infty$, the two marginals are equal.
\end{theorem}
In particular, if \eqref{eq:polynomial-connective-constant} holds with $\alpha=0$, this bound is uniform in $n$, so the family exhibits uniform strong spatial mixing at rate $Dc$.

\begin{remark}{(Why the polynomial prefactor matters.)}
Our definition of strong spatial mixing allows a polynomial prefactor, which differs from the standard notion of SSM.
The polynomial-prefactor statement need not imply uniform decay.
For example, let $G_h$ be the complete binary tree of height $h$.
Since $S_\ell(G_h)\le n$, \eqref{eq:polynomial-connective-constant} holds with $K=D=1$ and $\alpha=1$.
Consider four spins divided into two groups of size two, unit external fields, and an interaction equal to $5$ within a group and $1$ between groups.
Then $c_B=2/3$, so $Dc_B<1$.
If all leaves are pinned to a representative of one group or, alternatively, to a representative of the other group, the group magnetization propagated toward the root obeys
\[
    m_0=1,\qquad m_{j+1}=\frac{2c_Bm_j}{1+c_B^2m_j^2}.
\]
Consequently the total variation distance between the two root marginals tends to $\sqrt3/2$, rather than to zero.
The factor $n^\alpha$ in \eqref{eq:spatial-mixing-explicit} makes the asserted upper bound noninformative at the logarithmic root-to-leaf distance in this example.
\end{remark}

The same sufficient condition yields a deterministic approximation scheme
to the partition function $\Z_G(\boldsymbol \lambda, \boldsymbol A)$.
\begin{theorem}[Multi-spin systems] \label{thm:multi-spin-system}
    Fix an integer $q \ge 3$, and constants $K \ge 1, \alpha \ge 0, c > 0 $ and $D \ge 1$ such that $D  c < 1$.
    Let $\boldsymbol{\mathcal G} = (\mathcal G_n)_{n \ge 1}$ satisfy \eqref{eq:polynomial-connective-constant}. There exists a deterministic algorithm that, given an instance $\Phi = (G, [q], \boldsymbol A, \boldsymbol\lambda)$ with positive interactions, $G \in \mathcal G_n$ and $c_B(A_e) \le c$ for every $e \in E$, and a binary rational $\eps \in (0, 1)$, outputs an $\eps$-approximation to $\Z_G(\boldsymbol\lambda, \boldsymbol A)$ in 
    \[
    \biggl(\frac{n + \abs{\Phi} + \bit(\eps)}{\eps}\biggr)^C
    \]
    bit operations, where $C$ depends only on $q, K, \alpha, D$ and $c$.
\end{theorem}

For the $q$-state Potts model at edge interaction $\beta > 0$ with any external field $\boldsymbol\lambda$, all edges share the matrix $A$ with diagonal entries $\beta$ and off-diagonal entries $1$.
Then we have $\Delta_B(A) = 2\abs{\ln\beta}$ and $c_B(A) = \frac{\abs{\beta - 1}}{\beta + 1}$, which gives the following direct corollary.

\begin{corollary}[Potts model] \label{cor:general-Potts-model}
    Fix an integer $q \ge 3$, and constants $K \ge 1, \alpha \ge 0, D \ge 1$ and a positive rational number $\beta > 0$ such that $D \, \frac{\abs{\beta - 1}}{\beta + 1} < 1$.
    For every graph family $\boldsymbol{\mathcal G}=(\mathcal G_n)_{n\ge1}$ satisfying \eqref{eq:polynomial-connective-constant}, there is a deterministic algorithm that, given a Potts instance $\Phi$ on $G\in\mathcal G_n$ with positive rational external fields and a binary rational $\eps\in(0,1)$, outputs an $\eps$-approximation to its partition function in
    \[
    \biggl(\frac{n + \abs{\Phi} + \bit(\eps)}{\eps}\biggr)^C
    \]
    bit operations, where $C$ depends only on $q, K, \alpha, D$ and $\beta$.
\end{corollary}
When $\beta=1$, the partition function is exactly $\prod_{v\in V}\sum_{a\in[q]}\lambda_v(a)$; alternatively, \Cref{thm:multi-spin-system} applies with any fixed $c\in(0,1/D)$.
\begin{remark}
    \Cref{thm:multi-spin-system-spatial-mixing,thm:multi-spin-system} need the bounded connective constant of underlying graphs.
    We list some graph families satisfying this condition:
    \begin{enumerate}
        \item all $n$-vertex graphs of maximum degree $d\ge2$ have a polynomial connective constant $D=d-1$, with $K=d/(d-1)$ and $\alpha=0$; graphs of maximum degree at most $1$ satisfy the condition with $D=K=1$ and $\alpha=0$;
        \item for every fixed $d > 1$ and $\delta \in (0, 1)$, with high probability over $G$ drawn from the Erd\H{o}s-R\'{e}nyi random graph $\mathcal G(n, d/n)$, it has a polynomial connective constant $d(1 + \delta)$;
        \item some specific lattices have known polynomial connective constants; for example, finite subgraphs of the honeycomb lattice satisfy \eqref{eq:polynomial-connective-constant} for every fixed $D>\sqrt{2+\sqrt2}$, using the exact standard connective constant established in \cite{DCS12}.
    \end{enumerate}
\end{remark}

The second specific spin system is proper colorings on sparse Erd\H{o}s-R\'{e}nyi random graphs.
A proper $q$-coloring of $G = (V, E)$ is an assignment $\sigma : V \to [q]$ satisfying $\sigma(u) \neq \sigma(v)$ on every edge $(u, v) \in E$.
Denote by $\Omega(G, q)$ the collection of all proper $q$-colorings of $G$.
This is the spin system with $\lambda_v(c) = 1$ and $A_e(i, j) = \id{i \neq j}$.
The zero entries introduce feasibility constraints, so this application requires an extension of the positive-interaction construction.
\begin{theorem}[Colorings on random graphs] \label{thm:random-graph-coloring}
    For every $\eta \in (0, 1)$, there exists $d_0 = d_0(\eta)$ such that the following holds. For any fixed real $d \ge d_0$ and fixed integer $q \ge (2 + \eta)d$, there is a deterministic algorithm $\mathsf A = \mathsf A_{\eta, d, q}$ and a positive integer $n_0 = n_0(\eta, d, q)$ with the property: for every $n \ge n_0$, with probability $1 - o_n(1)$ over $G \sim \mathcal G(n, d/n)$, the graph is $q$-colorable and, for any binary rational $\eps \in (0, 1)$, on input $G$ and $\eps$, the algorithm $\mathsf A$ outputs an $\eps$-approximation to $\abs{\Omega(G, q)}$ in 
    \[
    \biggl(\frac{n + \bit(\eps)}{\eps}\biggr)^C
    \]
    bit operations, where $C$ depends only on $\eta, d$ and $q$.
\end{theorem}

\subsection{Technical overview}

\subsubsection*{Spatial mixing along self-avoiding walks}
The marginal-ratio recursion expresses the influence on a vertex through marginals in smaller instances.
For positive interactions, the Birkhoff contraction bounds the contribution of each child by its Hilbert distance multiplied by the coefficient of the corresponding edge.
Expanding these bounds produces a sum over self-avoiding walks.
The coefficients contribute at most $c^\ell$ along a walk of length $\ell$, while \eqref{eq:polynomial-connective-constant} bounds the number of walks by $Kn^\alpha D^\ell$.
Their combined contribution therefore decays at rate $Dc<1$, without requiring contraction of the sum of influences at every vertex.

\subsubsection*{Deterministic approximate counting} We design \textbf{FPTAS}es via estimating marginal probabilities.
At a query $(\Phi,v,t)$, the algorithm constructs a rational polytope $\polytope_t(\Phi,v)$ containing the true marginal vector at $v$.
The polytope enforces normalization and bounds on every pairwise marginal ratio.
For a fixed pair of spins $a,b$, the exact recursion has the form
\begin{align*}
    \frac{\mu_\Phi(v\gets a)}{\mu_\Phi(v\gets b)}
    =\frac{\lambda_v(a)}{\lambda_v(b)}
     \prod_i \frac{A_i(a,\cdot)\, p_i}{A_i(b,\cdot) \, p_i},
\end{align*}
where each child marginal $p_i$ depends on the pair $a,b$.
Each factor is a linear-fractional function of a single child marginal.
Minimizing and maximizing it over the child polytope gives certified bounds, whose products yield linear constraints on the parent marginal vector.
This optimization uses all constraints in the child polytope simultaneously.

The error measure is the Hilbert diameter of the feasible polytope.
Since the true marginal belongs to the polytope, a small diameter guarantees a multiplicative approximation from any feasible point.
For positive interactions, Birkhoff contraction and the self-avoiding-walk bound control the propagation of these diameters.
We round every computed lower bound downward and every upper bound upward.
This preserves containment and keeps the coefficient encoding lengths polynomial.
A logarithmic truncation depth and logarithmic rounding precision suffice for the desired accuracy; a separate bound on the full computation tree establishes polynomial running time.




\subsubsection*{Permissive blocks of random graph colorings}
For proper colorings, a locally proper partial assignment need not extend to a coloring of the graph.
We use the permissive blocks of \cite{Yin14,YZ16} to obtain a recursion on blocks instead of vertices, for which local feasibility suffices.
The polytope now describes the joint distribution of a block.
A lower bound on each feasible block probability ensures positivity, while upper bounds on the single-vertex marginals sharpen the contraction of the ratio factors.
Specifically, if the relevant marginal coordinates are at most $1/s$, the Hilbert contraction coefficient is at most $1/(s-1)$.

For a child root of original degree $k$, \cite[Theorem~7]{YZ16} gives the proper-coloring coefficient $2/(q-k-1)$ in its low-degree case.
Yin's spatial-mixing analysis obtains $1/(q-k-1)$ when comparing true marginals \cite{Yin14}.
At each nonexceptional child root in our computation, every distribution in the block polytope has projected color marginals at most $1/(q-k)$.
Optimizing $(1-p_a)/(1-p_b)$ over this common feasible set and applying \Cref{lem:coloring-Hilbert-contraction} yields the coefficient $1/(q-k-1)$ for the error of the computed enclosure.

The random graph analysis has two roles.
Weighted self-avoiding-walk estimates control the accumulated error, and bounds on exceptional components and weighted block walks control the total computation cost.
The latter is necessary because polynomial work at each block does not alone bound the number of recursive calls.
At the heuristic level, a typical step has about $d$ continuations and coefficient about $1/(q-d)$, so contraction asks for $d/(q-d)<1$, explaining the leading condition $q>2d$.
The formal proof uses \eqref{eq:coloring-path-parameter-estimate} and \Cref{lem:path-weight-concentration} for paths containing high-degree vertices.
\Cref{lem:coloring-accuracy-bound} follows one maximizing pair of block colorings at every error node, whereas the algorithm processes every such pair; \eqref{eq:coloring-block-cost} and \Cref{lem:block-walk-weight} account for all of these calls.
Together these estimates yield \Cref{thm:random-graph-coloring} at $q\ge(2+\eta)d$.

\subsection{Discussion}
We put here some discussions and future directions.

\subsubsection*{Strong spatial mixing and approximate counting}
Strong spatial mixing is closely related to approximate counting, but it does not always lead to a deterministic algorithm.
On bounded-degree graphs, the route in \cite[Theorem~4]{CFGZZ25} additionally assumes a permissive system, coupling independence, and a marginal lower bound, which, however, fails on unbounded-degree graphs.
On the other hand, our work provides $\mathbf{FPTAS}$es for multi-spin systems with positive interactions on graphs with bounded connective constant and proper colorings on sparse Erd\H{o}s-R\'{e}nyi random graphs asymptotically matching the current regime of strong spatial mixing on both models respectively.
This motivates the following more specific question.
\begin{conjecture}
    For multi-spin systems on a family of graphs satisfying some regularities, when the Gibbs distribution exhibits exponentially strong spatial mixing, there is a deterministic approximation scheme to the partition function with running time at most (quasi-)polynomial in the size of input.
\end{conjecture}
Our proof under $Dc<1$ controls contraction for every feasible message in each algorithmically generated polytope, not only for pairs of true conditional marginals.
How to derive such computable-message control directly from strong spatial mixing is the main unresolved step in the question.

\subsubsection*{Recovering earlier correlation-decay guarantees}
Our algorithm provides a route to retaining the guarantees of earlier recursive correlation-decay algorithms.
We expect the framework to recover, in principle, all previous results based on recursive correlation decay.
A polynomial-time recovery also requires that the recursion, any change of message coordinates, the enclosure representation, and outward rounding preserve the earlier accuracy and complexity guarantees.
The Birkhoff contraction in Hilbert metric used in our positive-interaction theorem is one analysis of the framework; recovering a different result may require its original potential function or other model-specific bounds.
The full coverage statement is a methodological expectation, rather than a universal simulation theorem established here.
Our coloring result demonstrates a concrete gain beyond the parameter range of the earlier recursive counting method of \cite{YZ16}.

\subsubsection*{Picture of approximate counting}
Besides the correlation-decay method, there are different ways to derive $\mathbf{FPTAS}$ for multi-spin systems: polynomial interpolations from zero-freeness \cite{LSS19,LSS25,BBR25}, marginal estimators by derandomizing Markov chain Monte Carlo methods \cite{FGWWY25}, Moitra's linear-programming methods on coupling independence \cite{CFGZZ25} and algorithmic methods from statistical physics \cite{HPR20,JKP20}.
It remains an open problem whether the algorithm proposed in this paper can cover or even surpass the results by other techniques.

\subsubsection*{Deterministic and randomized algorithms}
Randomized algorithms, such as Markov chain Monte Carlo, have been widely studied in approximation algorithms.
It has been commonly believed that deterministic algorithms could have the same capacities as randomized algorithms.
A recent work \cite{CFGZZ25} on Moitra's linear-programming methods on coupling independence provides several deterministic counting results matching sampling results for spin systems on graphs with bounded degree.
However, when the bounded-degree assumption is missing, there still remains a gap to reach current results of randomized algorithms. For example, \cite{EHSV18} shows the rapid mixing of Glauber dynamics of random graph colorings when $q > 1.7633 d$, while a recent work \cite{LY26} proves rapid mixing of Glauber dynamics as long as $q > (1 + o_d(1))d$; but in this work, we are only able to achieve $q \ge (2 + o_d(1))d$.
In spite of the gap, we hope that our algorithm can give some inspiration on this topic.

\subsubsection*{Turing machine vs. arithmetic models}
The computation model here is the deterministic Turing machine.
Rational linear-fractional programs can be solved in polynomial time in their binary encoding length. The existence of a strongly polynomial algorithm, whose number of arithmetic operations is polynomial in the numbers of variables and constraints and independent of coefficient bit lengths, remains open even for linear programming.

\subsection{Organization}
\Cref{sec:preliminaries} introduces the notation, Hilbert metric, random graphs, and some necessary tools.
Spatial mixing \Cref{thm:multi-spin-system-spatial-mixing} for positive interactions is given in \Cref{sec:spatial-mixing}, and our \textbf{FPTAS} and its approximation guarantee are provided in \Cref{sec:multi-spin-connective-constant}.
\Cref{sec:random-graph-coloring} extends our algorithm to proper colorings of sparse random graphs. The proofs of the Hilbert contraction inequalities appear in Appendix~\ref{sec:contraction-proof}, and the parameter verifications for \Cref{sec:random-graph-properties} appear in Appendix~\ref{sec:coloring-parameter-choice}.

\subsection{AI disclosure}
The main idea of the algorithm framework was provided by authors.
The proofs in this paper were generated after supplying ChatGPT Pro 5.6 Sol Ultra with an earlier research draft containing a preliminary algorithmic idea.
The agents suggested using the polynomial bit complexity of linear-fractional programming and provided parameter choices in \Cref{sec:random-graph-properties} and some calculations.
The authors have verified all proofs and take full responsibility for the paper.

\section{Preliminaries} \label{sec:preliminaries}
This section fixes basic notations and introduces necessary toolkits and preliminaries in the work.

\subsection{Mathematical notations}
Let $\mathbb R$ be all real numbers, $\mathbb R_{> 0}$ be all positive real numbers and $\mathbb N$ be all natural numbers.
We use $\e$ to denote the natural logarithm base and $\ln$ to denote the natural logarithm function.
For a natural number $n \in \mathbb N$, we use $[n]$ to denote the set $\set{1, \ldots, n}$.

For a graph $G$, $V(G)$ denotes its vertex set and $E(G)$ denotes its edge set.
When the context is clear, we simply use $V$ and $E$.
We give an order to $V$ and use the alphabetical order to compare paths.
The notation $\deg_G(v)$ is the degree of $v$ in $G$, \IE, the number of neighbors of $v$ in $G$.
For a vertex subset $B \subseteq V$, let $\partial_E B$ be the edge-boundary of $B$ in $G$.
We always list $\partial_E B = \set{e_1 = (u_1, v_1), \ldots, e_m = (u_m, v_m)}$ in a fixed order on $E$ such that $u_i \in B$ and $v_i \notin B$ for $i = 1, \ldots, m$.

Without any additional statement, we define a vector $u \in \mathbb R^d$ as a column vector in $\mathbb R^{d \times 1}$.
For a matrix $A \in \mathbb R^{d \times d}$, we use $A(i, \cdot) \in \mathbb R^{1 \times d}$ to denote the \emph{row} vector of $A$ at row $i$.
To take the value of an entry in a vector $u$ or a matrix $A$, the notations $u(i)$ and $A(i, j)$ are employed together with the standard notation $u_i$.

For an assignment (or coloring) $\sigma : U \to [q]$ and a subset $S \subseteq U$, we use $\sigma(S)$ to denote the partial assignment of $\sigma$ on $S$, \IE, $\sigma(S) = \sigma|_S$ and $\sigma(v)$ to denote $\sigma(\set{v})$ when $S = \set{v}$.
For a single element $v \in U$, we also use $\sigma(v)$ to denote the assigned value on $v$.
We use $v \gets c$ to describe the term that $v$ is assigned the value $c$ and use $B \gets \sigma$ to describe the term that the subset $B$ is assigned the partial assignment $\sigma$.

\subsection{Hilbert distance}
We introduce the following metric used by \cite{GK12,GKM15,Yin14,YZ16} to describe the distance of two positive vectors.
For two positive vectors $p, p'$ supported on a common state space $\Omega$, define the Hilbert distance between $p, p'$ as
\begin{align} \label{eq:Hilbert-distance}
    d_H(p, p') \defeq \max_{a \in \Omega} \ln\frac{p(a)}{p'(a)} - \min_{b \in \Omega} \ln\frac{p(b)}{p'(b)} = \max_{a, b \in \Omega} \ln\frac{p(a) p'(b)}{p'(a)p(b)}.
\end{align}
It is trivial to see that any normalization of vectors does not change the Hilbert distance.
The following inequality uses the Hilbert metric to bound the total variation distance.
\begin{lemma} \label{lem:TV-distance-by-Hilbert-distance}
    For two positive probability vectors $p, p'$ supported on a common state space $\Omega$, it holds that
    \begin{align*}
        \TV{p}{p'} \le \tanh\ab(\frac{d_H(p, p')}{4}) \le \frac{d_H(p, p')}{4}.
    \end{align*}
\end{lemma}
\begin{proof}
    When $p = p'$, the lemma is trivial.
    Otherwise, define the ratio vector of $p$ and $p'$ as $R(c) = p(c)/p'(c)$, and set $m = \min_{c \in \Omega} R(c)$ and $M = \max_{c \in \Omega} R(c)$.
    It is clear that $m \le 1 \le M$.
    We compute the total variation distance as
    \begin{align*}
        \TV{p}{p'} = \sum_{c \in \Omega} \max\set{p(c) - p'(c), 0} = \sum_{c \in \Omega} p'(c) \max\set{R(c) - 1, 0}.
    \end{align*}
    By the secant-line bound
    \begin{align*}
        \max \set{0, x - 1} \le \frac{M - x}{M - m} \max\set{0, m - 1} + \frac{x - m}{M - m} \max\set{0, M - 1} = \frac{(M - 1)(x - m)}{M - m},
    \end{align*}
    it holds that
    \begin{align*}
        \TV{p}{p'} \le \sum_{c \in \Omega} p'(c) \frac{(M - 1)(R(c) - m)}{M - m} = \frac{M - 1}{M - m} (1 - m).
    \end{align*}
    Observe that $M/m = \e^{t}$ where $t = d_H(p, p')$.
    By the AM-GM inequality $\e^t m + \frac{1}{m} \ge 2 \e^{t/2}$, we have
    \begin{align*}
        \TV{p}{p'} \le \frac{\e^t + 1 - (\e^tm + 1/m)}{\e^t - 1} \le \frac{\e^t - 2\e^{t/2} + 1}{\e^t - 1} = \frac{\e^{t/2} - 1}{\e^{t/2} + 1} = \tanh(t/4).
    \end{align*}
    This proves the first inequality.
    For the second inequality, just observe that $\tanh(x) \le x$ on $[0, \infty)$ and we conclude it.
\end{proof}
The following Birkhoff contraction for the Hilbert distance plays a crucial role.
It is actually a standard finite-dimensional argument in \cite{Birkhoff57} and for completeness, we include its proof in \Cref{sec:contraction-proof}.
\begin{lemma} \label{lem:Birknoff-contraction}
    For every positive symmetric matrix $A$ and positive vectors $x, y$, it holds that
    \begin{align*}
        d_H(Ax, Ay) \le c_B(A) d_H(x, y).
    \end{align*}
\end{lemma}

\subsection{Self-reducibility} \label{subsec:self-reducibility}
In the design of approximation algorithms, self-reducibility is an important property.
\cite{JVV86} shows that for a family of self-reducible instances, it suffices to estimate the marginal probability to approximate partition functions.
For simplicity, we list some self-reducible instances here instead of introducing the precise concept.
\begin{itemize}
    \item Multi-spin systems on graphs with bounded degree or with bounded connective constant are self-reducible.
    \item For colorings, pinning a vertex to a color removes that color from the lists of its unpinned neighbors. The residual instance is naturally a proper \emph{list-coloring} instance, and further feasible pinnings preserve this form. 
\end{itemize}

\subsection{Rounding length and rounding oracle}
To maintain the bit complexity, we introduce the rounding length and rounding oracles.
Let $P \ge 3$ be a picked rounding length and let $\xi_P = 2^{2 - P}$ be the rounding accuracy.
We define two oracles $\mathsf{round}_P^-(x)$ and $\mathsf{round}_P^+(x)$ on $(0, \infty)$ as follows.
Let $s$ be the unique integer satisfying $2^s \le x < 2^{s + 1}$ and set $h_P(x) = 2^{s - P + 1}$.
We define $\mathsf{round}_P^-(x) = h_P(x) \floor{\frac{x}{h_P(x)}}$ and $\mathsf{round}_P^+(x) = h_P(x) \ceil{\frac{x}{h_P(x)}}$.
It holds that the bit complexity of computing $\mathsf{round}_P^+(x)$ and $\mathsf{round}_P^-(x)$ is at most $O((\bit(x) + P)^2)$.
Furthermore, assume that $x \in [\ell, u]$ where $0<\ell\le u$ and both endpoints can be encoded by at most $L$ bits.
Then it takes at most $O(L + P)$ bits to encode the rounded rational numbers $\mathsf{round}_P^-(x)$ and $\mathsf{round}_P^+(x)$.
Meanwhile $0 < \mathsf{round}_P^-(x) \le x \le \mathsf{round}_P^+(x)$ and
\begin{align} \label{eq:rounding-accuracy}
    \max\; \set{\ln\frac{x}{\mathsf{round}_P^-(x)},\; \ln\frac{\mathsf{round}_P^+(x)}{x}} \le \xi_P.
\end{align}

\subsection{Linear-fractional programming}
Now we state some standard results for linear-fractional programs.
A rational linear-fractional program $\mathcal L$ is of the form
\begin{equation} \label{eq:linear-fractional-programming}
\begin{aligned}
    \max/\min \; &\frac{\mathbf c^\top \mathbf x + \alpha}{\mathbf d^\top \mathbf x + \beta} \\
    \text{subject to} \; & A \mathbf x \le \mathbf b.
\end{aligned}
\end{equation}
Here $\mathbf x\in\mathbb R^{n_{\mathrm{var}}}$ is the vector of variables; $\mathbf c,\mathbf d\in\mathbb Q^{n_{\mathrm{var}}}$ and $\alpha,\beta\in\mathbb Q$ specify the objective; and $A\in\mathbb Q^{m_{\mathrm{con}}\times n_{\mathrm{var}}}$ and $\mathbf b\in\mathbb Q^{m_{\mathrm{con}}}$ specify the constraints. Write $|\mathcal L|$ for the total binary encoding length of these data.

The Charnes--Cooper transformation \cite{CC62} converts such a program to a linear program when its denominator is positive throughout the feasible region.
Polynomial-time solvability in the binary model is a separate result of linear programming \cite{Khachiyan79,Karmarkar84}.
\begin{theorem}[Linear-fractional optimization \cite{CC62,Khachiyan79,Karmarkar84}] \label{thm:Charness-Cooper-transformation}
    Suppose that the feasible region of \eqref{eq:linear-fractional-programming} is a nonempty bounded rational polytope and that $\mathbf d^\top\mathbf x+\beta>0$ for every feasible $\mathbf x$.
    Then $\mathcal L$ can be transformed to a rational linear program of encoding length $O(|\mathcal L|)$, with $O(n_{\mathrm{var}})$ variables and $O(m_{\mathrm{con}})$ constraints.
    Its exact rational optimum can be computed within the conservative bound
    \[
        O\bigl((n_{\mathrm{var}}+m_{\mathrm{con}})^4|\mathcal L|^3\bigr)
    \]
    bit operations.
\end{theorem}
\begin{proof}
    Set $t=(\mathbf d^\top\mathbf x+\beta)^{-1}$ and $\mathbf y=t\mathbf x$.
    The transformed objective is $\mathbf c^\top\mathbf y+\alpha t$, subject to $A\mathbf y\le\mathbf b t$, $\mathbf d^\top\mathbf y+\beta t=1$, and $t\ge0$.
    Boundedness of the original feasible region implies that its recession cone is $\{0\}$, so a feasible transformed point cannot have $t=0$.
    Thus the inverse map $\mathbf x=\mathbf y/t$ preserves feasible points and objective values.
    The coefficient encoding grows by only a constant factor.
    Introducing slack variables and splitting free variables gives $N=O(n_{\mathrm{var}}+m_{\mathrm{con}})$ variables in standard form.
    Rational linear programming \cite{Karmarkar84}, including exact recovery of an optimum, takes $O(N^{3.5}|\mathcal L|)$ arithmetic operations on $O(|\mathcal L|)$-bit numbers.
    Charging quadratic bit cost for each operation yields the displayed coarse bound.
\end{proof}

\subsection{Erd\H{o}s-R\'{e}nyi random graphs}
The Erd\H{o}s-R\'{e}nyi graph $G = (V, E) \sim \mathcal G(n, p)$ on $n$ vertices of rate $p \in [0, 1]$ is a random graph generated as:
\begin{itemize}
    \item $V = [n]$; and
    \item for every $(i, j) \in \binom{[n]}{2}$, $E$ includes it with probability $p$ independently.
\end{itemize}
The following lemma gives an algorithm to find a proper $q$-coloring on random graphs.
\begin{lemma} \label{lem:random-graph-colorability}
    For every $d \ge 1$ and $q > \max\;\set{d, 3}$, with probability $1 - o_n(1)$ over $G \sim \mathcal G(n, d/n)$, $G$ is $q$-colorable and a proper $q$-coloring on $G$ can be output in $\poly{n, q}$ bit operations.
\end{lemma}
\begin{proof}
    By \cite{PSW96}, with probability $1 - o_n(1)$, $G$ has no $q$-core, \IE, an induced subgraph of minimal degree at least $q$.
    Consider the following algorithm:
    \begin{enumerate}
        \item initialize $H = G$, $S = \emptyset$ and $Q = \set{v \in V(H) \;:\; \deg_H(v) \le q - 1}$;
        \item repeat the following process until $Q = \emptyset$:
        \begin{enumerate}
            \item pick a smallest vertex $v \in Q$;
            \item append $v$ to the tail of $S$ and remove $v$ from $Q$;
            \item remove $v$ from $H$ and delete all incident edges to $v$;
            \item $Q \gets Q \cup \set{u \in V(H) \;: \; \deg_H(u) \le q - 1}$;
        \end{enumerate}
        \item from the tail of $S$ to the head, color the vertex with the smallest color which is not used by assigned neighbors.
    \end{enumerate}
    When $G$ has no $q$-core, it holds that when $Q = \emptyset$, $H$ is empty and thus $S$ contains all vertices in $V(G)$.
    The feasibility of the coloring follows directly from the order used to construct $S$.
    For completeness, a straightforward implementation may rescan all surviving vertices and their incident-edge lists in each of at most $n$ deletion rounds, and may scan incident edges while trying at most $q$ colors during reverse coloring.
    Since vertex labels, degrees, and colors use $O(\log(n+q))$ bits, this gives the coarse bound
    \[
        O\bigl((n+q)(n+\abs{E(G)})\log(n+q)\bigr)
    \]
    bit operations, which is $\poly{n,q}$ as required.
\end{proof}

\section{Spatial Mixing of Multi-spin Systems} \label{sec:spatial-mixing}
In this section, we show the strong spatial mixing of multi-spin systems \Cref{thm:multi-spin-system-spatial-mixing}.
The key ingredient is the \emph{ratio recursion} of spin systems, which has been widely used in the analysis of proper coloring models \cite{GK12,GKM15}.

To describe the recursion, we first introduce the concept of pinnings.
\begin{definition} \label{def:pinnings}
    Given an instance $\Phi = (G = (V, E), [q], \boldsymbol A = (A_e)_{e \in E}, \boldsymbol\lambda = (\lambda_v)_{v \in V})$ with positive interactions and a vertex $v \in V$, assume that its neighbors are $v_1, \ldots, v_d$ where $d = \deg_G(v)$.
    For any colors $a, b \in [q]$ and $i \in [d]$, define the child instance $\Phi_{i, a, b}^v = (G', [q], \boldsymbol A', \boldsymbol \lambda')$ as:
    \begin{itemize}
        \item the underlying graph $G' = G[V \setminus \set v]$ is the induced subgraph of $G$ by removing $v$;
        \item $\boldsymbol A' = (A_e)_{e \in E \setminus \set{(v, v_1), \ldots, (v, v_d)}}$;
        \item $\boldsymbol\lambda' = (\lambda'_u)_{u \in V \setminus \set v}$ is defined as: for $j = 1, \ldots, i - 1$, $\lambda'_{v_j}(c) = A_{(v, v_j)}(a, c) \lambda_{v_j}(c)$; for $j = i + 1, \ldots, d$, $\lambda'_{v_j}(c) = A_{(v, v_j)}(b, c) \lambda_{v_j}(c)$; and $\lambda'_u = \lambda_u$ for remaining vertices.
    \end{itemize}
    Let $p_{i, a, b}^v$ be the marginal vector of $v_i$ in $\Phi_{i, a, b}^v$, \IE, $p_{i, a, b}^v(c) = \mu_{\Phi_{i, a, b}^v}(v_i \gets c)$ for every $c \in [q]$.
    When the context is clear, we simply use $\Phi_i$ to denote $\Phi_{i, a, b}^v$ and $p_i$ to denote $p_{i, a, b}^v$.
\end{definition}

Based on \Cref{def:pinnings}, we have the following recursion for the marginal ratio, which plays a central role in the proof of \Cref{thm:multi-spin-system-spatial-mixing,thm:multi-spin-system}.
\begin{lemma} \label{lem:spin-ratio-recursion}
    Given an instance $\Phi = (G = (V, E), [q], \boldsymbol A, \boldsymbol\lambda)$ with positive interactions, for every vertex $v \in V$ and $a, b \in [q]$, it holds that
    \begin{align*}
        \frac{\mu_\Phi(v \gets a)}{\mu_\Phi(v \gets b)} = \frac{\lambda_v(a)}{\lambda_v(b)} \prod_{i = 1}^d \frac{\sum_{c \in [q]} A_{(v, v_i)}(a, c) \mu_{\Phi_i}(v_i \gets c)}{\sum_{c \in [q]} A_{(v, v_i)}(b, c) \mu_{\Phi_i}(v_i \gets c)}.
    \end{align*}
\end{lemma}
\begin{proof}
    The proof is similar to \cite[Proposition 1]{LY13}.
    For a partial assignment $\pi$ on $V \setminus \set v$, define
    \begin{align*}
        W(\pi) \defeq \prod_{u \in V \setminus \set v} \lambda_u(\pi(u)) \prod_{(u, w) \in E \setminus \set{(v, v_1), \ldots, (v, v_d)}} A_{(u, w)}(\pi(u), \pi(w)).
    \end{align*}
    For every $c \in [q]$, let
    \begin{align*}
        F(c) \defeq \sum_{\pi : V \setminus \set v \to [q]} W(\pi) \prod_{i = 1}^d A_{(v, v_i)}(c, \pi(v_i)).
    \end{align*}
    By \eqref{eq:spin-Gibbs-distribution},
    \begin{align*}
        \frac{\mu_\Phi(v \gets a)}{\mu_\Phi(v \gets b)} = \frac{\lambda_v(a) F(a)}{\lambda_v(b) F(b)}.
    \end{align*}
    To telescope $F(a)/F(b)$, for $i = 0, \ldots, d$, define
    \begin{align} \label{eq:spin-telescoping}
        F_i \defeq \sum_{\pi : V \setminus \set v \to [q]} W(\pi) \prod_{j = 1}^i A_{(v, v_j)}(a, \pi(v_j)) \prod_{j = i + 1}^d A_{(v, v_j)}(b, \pi(v_j)).
    \end{align}
    Then it holds that $F(a)/F(b) = \prod_{i = 1}^d F_i/F_{i - 1}$ since $F(a) = F_d$ and $F(b) = F_0$.
    For a fixed $i = 1, \ldots, d$, write
    \begin{align*}
        W_i(\pi) = W(\pi) \prod_{j = 1}^{i - 1} A_{(v, v_j)}(a, \pi(v_j)) \prod_{j = i + 1}^d A_{(v, v_j)}(b, \pi(v_j)).
    \end{align*}
    Observe that the partition function of $\Phi_i$ is equal to
    \begin{align*}
        \Z(\Phi_i) = \sum_{\pi : V \setminus \set v \to [q]} W_i(\pi)
    \end{align*}
    and for the marginal vector $p_i$ on $v_i$,
    \begin{align*}
        \forall c \in [q], \quad p_i(c) = \frac1{\Z(\Phi_i)} \sum_{\pi : V \setminus \set v \to [q], \pi(v_i) = c} W_i(\pi).
    \end{align*}
    Then we have
    \begin{align*}
        F_i &= \sum_{\pi : V \setminus \set v \to [q]} A_{(v, v_i)}(a, \pi(v_i)) W_i(\pi) \\
        &= \Z(\Phi_i) \sum_{c \in [q]} A_{(v, v_i)}(a, c) p_i(c)
    \end{align*}
    and similarly $F_{i - 1} = \Z(\Phi_i) \sum_{c \in [q]} A_{(v, v_i)}(b, c) p_i(c)$.
    Therefore
    \begin{align*}
        \frac{F_i}{F_{i - 1}} = \frac{\sum_{c \in [q]} A_{(v, v_i)}(a, c) p_i(c)}{\sum_{c \in [q]} A_{(v, v_i)}(b, c) p_i(c)}.
    \end{align*}
    Substituting this into the telescoping product $F(a)/F(b) = \prod_{i = 1}^d F_i/F_{i - 1}$ proves the lemma.
\end{proof}

\begin{proof}[Proof of \Cref{thm:multi-spin-system-spatial-mixing}]
    Fix the boundary set $\Lambda$ and two different partial assignments $\sigma, \tau$ on $\Lambda$.
    Let $\Gamma = \Delta(\sigma, \tau)$, \IE, the discrepancy of $\sigma$ and $\tau$.
    Set $\overline \Delta = 4\arctanh c$ and $\rho = D \, c$.
    For convenience, denote by $H_\Phi(v, \Lambda, \sigma, \tau)$ the Hilbert distance $d_H(\mu_\Phi(v \gets \cdot \;|\; \Lambda \gets \sigma), \mu_\Phi(v \gets \cdot \;|\; \Lambda \gets \tau))$.

    Write $r=\dist_G(v,\Gamma)$. If $r=\infty$, the connected component of $v$ contains no boundary vertex on which $\sigma$ and $\tau$ differ, so the two conditional marginals at $v$ are identical. Assume henceforth that $r<\infty$.
    We will show that
    \begin{align} \label{eq:spatial-mixing-Hilbert-decay}
        H_\Phi(v, \Lambda, \sigma, \tau) \le \frac{K n^\alpha D \overline{\Delta}}{1 - \rho} \rho^{r - 1}.
    \end{align}
    Therefore, by \Cref{lem:TV-distance-by-Hilbert-distance},
    \[
        \TV{\mu_\Phi(v\gets\cdot\mid\Lambda\gets\sigma)}
            {\mu_\Phi(v\gets\cdot\mid\Lambda\gets\tau)}
        \le \frac{KD\overline\Delta}{4\rho(1-\rho)}\,n^\alpha\rho^r,
    \]
    which is \eqref{eq:spatial-mixing-explicit}.
    
    We first introduce \emph{local activities} on $V \setminus \Lambda$ after a pinning $\pi$ on $\Lambda$.
    Let
    \begin{align*}
        \forall x \in V \setminus \Lambda, c \in [q], \quad \lambda_x^\pi(c) = \lambda_x(c) \prod_{y \in \Lambda \cap N_G(x)} A_{(x, y)}(c, \pi(y)).
    \end{align*}
    The local Hilbert discrepancy is induced by local activities conditional on different pinnings:
    \begin{align*}
        H^\local(x) \defeq d_H(\lambda_x^\sigma, \lambda_x^\tau).
    \end{align*}
    By direct calculation on $d_H(\lambda_x^\sigma, \lambda_x^\tau)$, it holds that
    \begin{align*}
        H^\local(x) \le \sum_{y \in \Gamma \cap N_G(x)} \Delta_B\ab(A_{(x, y)}).
    \end{align*}
    We remark here $H^\local(x) = 0$ unless $x$ is adjacent to $\Gamma$.

    To show \eqref{eq:spatial-mixing-Hilbert-decay}, we specify a sub-instance $\mathcal I$ by a vertex subset $S = S_{\mathcal I}$ and an external field $\boldsymbol h = (h_x : [q] \to \mathbb R_{> 0})_{x \in S}$ with the original edge interactions restricted on $G[S]$.
    We say a pair of instances $(\mathcal I = (S, \boldsymbol h), \mathcal I' = (S, \boldsymbol h'))$ are \emph{admissible} if $S \subseteq V \setminus \Lambda$ and for every $x \in S$,
    \begin{align*}
        \forall c \in [q], \quad \frac{h_x(c)}{h_x'(c)} = \frac{\lambda_x^\sigma(c)}{\lambda_x^\tau(c)}.
    \end{align*}
    Consider the Hilbert distance $d_H(\mu_{\mathcal I}(x \gets \cdot), \mu_{\mathcal I'}(x \gets \cdot))$.
    Assume that $x_1, \ldots, x_d$ are the neighbors of $x$ in $G[S]$.
    For two positive probability vectors $p, p'$ and $a, b \in [q]$, by \Cref{lem:Birknoff-contraction},
    \begin{align*}
        \abs{\ln\frac{\sum_{c \in [q]} A_{(x, x_i)}(a, c) p(c)}{\sum_{c \in [q]} A_{(x, x_i)}(b, c) p(c)} - \ln\frac{\sum_{c \in [q]} A_{(x, x_i)}(a, c) p'(c)}{\sum_{c \in [q]} A_{(x, x_i)}(b, c) p'(c)}} &\le d_H\ab(A_{(x, x_i)}p, A_{(x, x_i)}p') \\
        &\le c_B\ab(A_{(x, x_i)}) d_H(p, p').
    \end{align*}
    Apply \Cref{lem:spin-ratio-recursion}, and we obtain that
    \begin{align*}
        d_H(\mu_{\mathcal I}(x \gets \cdot), \mu_{\mathcal I'}(x \gets \cdot)) \le H^\local(x) + \max_{a, b \in [q]} \sum_{i = 1}^d c_B\ab(A_{(x, x_i)}) d_H(\mu_{\mathcal I_i}(x_i \gets \cdot), \mu_{\mathcal I'_i}(x_i \gets \cdot))
    \end{align*}
    where $\mathcal I_i, \mathcal I'_i$ are generated by \Cref{def:pinnings} from $\mathcal I$, $\mathcal I'$ respectively.

    For a vertex subset $S$ and a vertex $x \in S$, introduce the function $\Inf(S, x)$ in the following recursive form:
    \begin{align} \label{eq:spatial-mixing-recursive-bound}
        \Inf(S, x) = H^\local(x) + \sum_{i = 1}^d c_B\ab(A_{(x, x_i)}) \Inf(S \setminus \set{x}, x_i)
    \end{align}
    with boundary condition $\Inf(S, x) = H^\local(x)$ for $S = \set{x}$ or $x$ is isolated in the graph $G[S]$.
    By induction, $d_H(\mu_{\mathcal I}(x \gets \cdot), \mu_{\mathcal I'}(x \gets \cdot))$ is upper bounded by $\Inf(S, x)$ for every pair of admissible $(\mathcal I, \mathcal I')$.
    In addition, note that the instance $\Phi$ with pinning $\pi$ on $\Lambda$ is specified by $V \setminus \Lambda$ and $\boldsymbol{\lambda}^\pi = (\lambda_x^\pi)_{x \in V \setminus \Lambda}$.
    Therefore, it holds that
    \begin{align*}
        H_\Phi(v, \Lambda, \sigma, \tau) \le \Inf(V \setminus \Lambda, v).
    \end{align*}
    Unfold \eqref{eq:spatial-mixing-recursive-bound} and substitute the bound of $H^\local(x)$.
    Every non-zero term is encoded by a self-avoiding walk $P = (v = x_0, \ldots, x_{\ell - 1})$ in $G[V \setminus \Lambda]$ and a vertex $y \in \Gamma$ incident to $x_{\ell - 1}$.
    Append $y$ to the tail of $P$.
    Therefore $P$ is a self-avoiding walk from $v$ to $\Gamma$ of length $\ell$ with contribution at most $c^{\ell - 1} \cdot \overline{\Delta}$.
    Hence,
    \begin{align*}
        \Inf(V \setminus \Lambda, v) &\le \sum_{\ell = \dist_G(v, \Gamma)}^{n - 1} c^{\ell - 1} \cdot \overline{\Delta} \cdot \abs{\SAW_G(v, \ell)} \\
        &\le K n^\alpha D \cdot \overline{\Delta} \sum_{\ell \ge \dist_G(v, \Gamma) - 1} \rho^\ell \\
        &= \frac{K n^\alpha D \overline \Delta}{1 - \rho} \rho^{\dist_G(v, \Gamma) - 1}.
    \end{align*}
    Combining all things together, we conclude the theorem.
\end{proof}

\section{Approximate Counting Multi-spin Systems with Positive Interactions} \label{sec:multi-spin-connective-constant}
In this section, we design the algorithm in \Cref{thm:multi-spin-system}.
Based on the recursion in \Cref{sec:spatial-mixing}, we show how to recursively construct rational polytopes which approximately estimate the true marginal vector and analyze them when the instance is defined on a graph of bounded connective constant in \Cref{subsec:spin-polytope-construction} and we derive the approximation scheme to the partition function by polytopes in \Cref{subsec:spin-approximation-scheme}.

\subsection{Rational polytopes induced by the recursion} \label{subsec:spin-polytope-construction}
Recall \Cref{lem:spin-ratio-recursion}.
Consider a query $(\Phi = (G = (V, E), [q], \boldsymbol A, \boldsymbol\lambda), v, t)$ for an instance $\Phi$ with positive rational interactions and external fields, a vertex $v \in V$ and an integer $t \ge 0$ representing the recursive depth.
We recursively construct a rational polytope $\polytope = \polytope_t(\Phi, v)$ to estimate the marginal vector $\mu_\Phi(v \gets \cdot)$ as \Cref{def:spin-polytope}.
\begin{definition}[Rational polytope for spin systems] \label{def:spin-polytope}
    The variables of $\polytope_t(\Phi, v)$ are $(z_c)_{c \in [q]}$.
    When $v$ is an isolated vertex in $G$, $\polytope_t(\Phi, v)$ is the exact polytope
    \begin{align*}
        \forall c \in [q], \quad z_c = \frac{\lambda_v(c)}{\sum_{k \in [q]} \lambda_v(k)}.
    \end{align*}
    
    When $v$ is not isolated, assume that its neighbors in $G$ are $v_1, \ldots, v_d$ where $d = \deg_G(v)$.
    The polytope $\polytope_t(\Phi, v)$ at first contains:
    \begin{itemize}
        \item \textbf{Non-negativity.}
        \begin{align} \label{eq:spin-LFP-non-negativity}
            \forall c \in [q], \quad z_c \ge 0.
        \end{align}
        \item \textbf{Normalization.}
        \begin{align} \label{eq:spin-LFP-normalization}
            \sum_{c \in [q]} z_c = 1.
        \end{align}
    \end{itemize}
    
    We separate cases based on the value of the depth $t$.
    \begin{itemize}
        \item $t = 0$.
        For every pair of distinct spins $a, b \in [q]$, the polytope contains
        \begin{align} \label{eq:spin-LFP-depth-zero}
            L_{a, b}^0 \cdot z_b \le z_a \le U_{a, b}^0 \cdot z_b
        \end{align}
        where
        \begin{align*}
            L_{a, b}^0 \defeq \frac{\lambda_v(a)}{\lambda_v(b)} \prod_{i = 1}^d \min_{c \in [q]}\; \frac{A_{(v, v_i)}(a, c)}{A_{(v, v_i)}(b, c)}, \quad
            U_{a, b}^0 \defeq \frac{\lambda_v(a)}{\lambda_v(b)} \prod_{i = 1}^d \max_{c \in [q]}\; \frac{A_{(v, v_i)}(a, c)}{A_{(v, v_i)}(b, c)}.
        \end{align*}

        \item $t > 0$.
        For every pair of distinct spins $a, b \in [q]$ and $i = 1, \ldots, d$, construct an instance $\Phi_i = \Phi_{i, a, b}^v$.
        Recursively obtain the polytope $\polytope_i = \polytope_{t - 1}(\Phi_i, v_i)$.
        Solve the following two LFPs:
        \begin{align*}
            \ell_{i, a, b}^* \defeq \min_{z \in \polytope_i}\; \frac{\sum_{c \in [q]} A_{(v, v_i)}(a, c) z_c}{\sum_{c \in [q]} A_{(v, v_i)}(b, c) z_c}, \quad
            u_{i, a, b}^* \defeq \max_{z \in \polytope_i}\; \frac{\sum_{c \in [q]} A_{(v, v_i)}(a, c) z_c}{\sum_{c \in [q]} A_{(v, v_i)}(b, c) z_c},
        \end{align*}
        and set $\underline\ell_{i, a, b} = \mathsf{round}_P^-(\ell_{i, a, b}^*), \overline u_{i, a, b} = \mathsf{round}_P^+(u_{i, a, b}^*)$.
        Setting $L_{a, b}^t = \frac{\lambda_v(a)}{\lambda_v(b)} \prod_{i = 1}^d \underline\ell_{i, a, b}$ and $U_{a, b}^t = \frac{\lambda_v(a)}{\lambda_v(b)} \prod_{i = 1}^d \overline u_{i, a, b}$, the polytope then contains
        \begin{align} \label{eq:spin-LFP-ratio-recursion}
            L_{a, b}^t \cdot z_b \le z_a \le U_{a, b}^t \cdot z_b.
        \end{align}
    \end{itemize}
\end{definition}

\begin{fact}[Well-definedness and soundness of polytopes] \label{fact:spin-polytope-soundness}
    At every depth the recursive construction is well-defined, the true marginal vector at $v$ belongs to $\polytope_t(\Phi, v)$ and the bit complexity to encode $\polytope_t(\Phi, v)$ is at most $O(q^3(\abs\Phi + nP))$.
\end{fact}
\begin{proof}
    We prove the fact by induction.
    When the polytope contains only an exact point (the case $v$ is isolated), the fact is trivial.
    Otherwise, at depth $0$, it is direct to see the true marginal vector $p_v$ at $v$ satisfies constraints \eqref{eq:spin-LFP-non-negativity} and \eqref{eq:spin-LFP-normalization}.
    By \Cref{lem:spin-ratio-recursion}, it is easy to verify $p_v$ satisfies \eqref{eq:spin-LFP-depth-zero}.
    The bit complexity to encode $\polytope_0(\Phi, v)$ is at most $O(q^3 \abs \Phi)$ since we have $q$ variables and $O(q^2)$ constraints with each constraint encoded with at most $O(\abs\Phi)$ bits.

    Suppose that the assertion holds through depth $t - 1$.
    It is trivial to see the true marginal vector $p_v$ satisfies \eqref{eq:spin-LFP-non-negativity} and \eqref{eq:spin-LFP-normalization}.
    Since the true marginal vector $p_i$ belongs to $\polytope_i$ by induction, by \Cref{lem:spin-ratio-recursion}, $p_v$ also satisfies the constraint \eqref{eq:spin-LFP-ratio-recursion} even after applying rounding oracles.
    Therefore we conclude that $p_v \in \polytope_t(\Phi, v)$.
    To compute the number of bits to encode $\polytope_t(\Phi, v)$, note that we have $q$ variables and $O(q^2)$ linear constraints.
    For a constraint \eqref{eq:spin-LFP-ratio-recursion}, note that for each $i = 1, \ldots, d$, it holds that
    \begin{align*}
        \min_{c \in [q]} \frac{A_{(v, v_i)}(a, c)}{A_{(v, v_i)}(b, c)} \le \ell_{i, a, b}^* \le u_{i, a, b}^* \le \max_{c \in [q]} \frac{A_{(v, v_i)}(a, c)}{A_{(v, v_i)}(b, c)}.
    \end{align*}
    Hence, it takes $O(\bit(A_{(v, v_i)}) + P)$ bits to encode the rounded rational number $\underline \ell_{i, a, b}, \overline u_{i, a, b}$ and thus $O(\bit(\lambda_v) + dP + \sum_{i = 1}^d \bit(A_{(v, v_i)}))$ bits to encode $L_{a, b}^t$ and $U_{a, b}^t$.
    Therefore, the bit complexity to encode $\polytope_t(\Phi, v)$ is at most $O(q^3(\abs\Phi + nP))$.
\end{proof}

For an admissible polytope $\polytope$, we use $\diam_H(\polytope)$ to denote the Hilbert diameter of $\polytope$, \IE,
\begin{align*}
    \diam_H(\polytope) \defeq \max_{x, y \in \polytope} d_H(x, y).
\end{align*}
Based on the Birkhoff contraction \Cref{lem:Birknoff-contraction}, the following recursion is on the control of the accuracy in $\polytope_t(\Phi, v)$ by its child polytopes.
\begin{proposition} \label{prop:spin-accuracy-recursion}
    Fix an admissible feasible query $(\Phi, v, t)$.
    When $v$ is isolated, $\diam_H(\polytope_t(\Phi, v)) = 0$.
    Otherwise, assume that its neighbors are $v_1, \ldots, v_d$.
    It holds that
    \begin{align*}
        \diam_H(\polytope_t(\Phi, v)) \le \begin{cases}
            \sum_{i = 1}^d \Delta_B(A_{(v, v_i)}), & t = 0 \\
            \max_{a, b \in [q],\,a \ne b} \sum_{i = 1}^d \ab(c_B(A_{(v, v_i)}) \diam_H(\polytope_{t - 1}(\Phi_{i, a, b}^v, v_i)) + 2\xi_P), & t > 0
        \end{cases}\;.
    \end{align*}
\end{proposition}
\begin{proof}
    When $v$ is isolated, $\polytope_t(\Phi, v)$ contains only a single point and thus $\diam_H(\polytope_t(\Phi, v)) = 0$.
    Suppose that $v$ is not isolated and its neighbors are $v_1, \ldots, v_d$.
    The positive ratio bounds and normalization force every coordinate of every feasible point to be positive.
    For $t = 0$ and distinct $a,b\in[q]$, it holds that
    \begin{align*}
        \ln \frac{U_{a, b}^0}{L_{a, b}^0} = \sum_{i = 1}^d \ab(\ln \max_{c, c' \in [q]} \frac{A_{(v, v_i)}(a, c)A_{(v, v_i)}(b, c')}{A_{(v, v_i)}(a, c')A_{(v, v_i)}(b, c)}) \le \sum_{i = 1}^d \Delta_B(A_{(v, v_i)}).
    \end{align*}
    Then for every $x, y \in \polytope_0(\Phi, v)$, it holds that
    \begin{align*}
        d_H(x, y) = \max_{a, b \in [q],\,a \ne b} \ln \frac{x(a)y(b)}{x(b)y(a)} \le \max_{a, b \in [q],\,a \ne b} \ln\frac{U_{a, b}^0}{L_{a, b}^0} \le \sum_{i = 1}^d \Delta_B(A_{(v, v_i)}).
    \end{align*}
    For $t > 0$, for every $i = 1, \ldots, d$, by \Cref{lem:Birknoff-contraction},
    \begin{align*}
        \ln\frac{u_{i, a, b}^*}{\ell_{i, a, b}^*} &\le \max_{x, y \in \polytope_{t - 1}(\Phi_i, v_i)} d_H(A_{(v, v_i)}x, A_{(v, v_i)}y) \\
        &\le c_B(A_{(v, v_i)}) \diam_H(\polytope_{t - 1}(\Phi_i, v_i)).
    \end{align*}
    Together with \eqref{eq:rounding-accuracy}, it holds that
    \begin{align*}
        \ln\frac{U_{a, b}^t}{L_{a, b}^t} \le \sum_{i = 1}^d \ab(c_B(A_{(v, v_i)}) \diam_H(\polytope_{t - 1}(\Phi_{i, a, b}^v, v_i)) + 2\xi_P).
    \end{align*}
    A similar argument to the case $t = 0$ shows
    \begin{align*}
        \diam_H(\polytope_t(\Phi, v)) \le \max_{a, b \in [q],\,a \ne b} \ln \frac{U_{a, b}^t}{L_{a, b}^t}
    \end{align*}
    and thus we conclude the full proposition.
\end{proof}

The following bound of accuracy is a corollary of \Cref{prop:spin-accuracy-recursion}.
\begin{corollary} \label{cor:spin-accuracy}
    Fix constants $K \ge 1$, $\alpha \ge 0$, $D \ge 1$ and $c \in (0,1/D)$.
    Given an instance $\Phi = (G = (V, E), [q], \boldsymbol A, \boldsymbol\lambda)$ with positive rational interactions and external fields such that $G$ satisfies \eqref{eq:polynomial-connective-constant} and $c_B(A_e) \le c$ for every $e \in E$, for every integer $R \ge 0$ and $v \in V$, it holds that
    \begin{align} \label{eq:spin-accuracy-bound}
        \diam_H(\polytope_R(\Phi, v)) \le K n^\alpha D \ab(\overline\Delta \rho^R + \frac{2\xi_P}{1 - \rho})
    \end{align}
    where $\overline\Delta = 4\arctanh c$ and $\rho = D \cdot c \in (0,1)$.
    The same fixed $c$ is used for all derived instances, including those with no edges or with $c_B(A_e)=0$ on every edge.
    Moreover, for every $\eps \in (0, 1)$, set
    \begin{align} \label{eq:spin-depth-accuracy-parameters}
        R = R_\eps = \left\lceil \frac{\ln(4Kn^\alpha D\max\set{1, \overline\Delta}/\eps)}{\ln(1/\rho)}\right\rceil, \quad
        P = P_\eps = \max \; \set{3, \left\lceil \log_2 \frac{32Kn^\alpha D}{\eps(1 - \rho)}\right\rceil}.
    \end{align}
    For every feasible point $z \in \polytope_R(\Phi, v)$, it holds that for every $c \in [q]$,
    \begin{align*}
        \e^{-\eps} \mu_\Phi(v \gets c) \le z_c \le \e^\eps \mu_\Phi(v \gets c).
    \end{align*}
\end{corollary}
\begin{proof}
    The proof is similar to \Cref{thm:multi-spin-system-spatial-mixing}.
    Fix an instance $\Phi = (G = (V, E), [q], \boldsymbol A, \boldsymbol \lambda)$, a vertex $v \in V$ and a depth $R \ge 0$.
    First we have
    \begin{align*}
        \Delta_B(A_e) = 4\arctanh c_B(A_e) \le 4\arctanh c = \overline\Delta.
    \end{align*}
    
    Consider the recursion in \Cref{prop:spin-accuracy-recursion}.
    We unfold the recursion into a rooted tree $\mathcal T$ with root $\varrho$.
    A node $\omega$ in $\mathcal T$ at level $0 \le r \le R$ is labelled by a derived instance $\Phi_\omega$ on the underlying graph $G_\omega$, a vertex $v_\omega$ and a remaining depth $R - r$.
    At every node $\omega$ at level $r < R$, we choose any pair of $(a_\omega, b_\omega)$ to attain the maximum in the recursion.
    Then for each vertex $u$ incident to $v_\omega$ in the instance $\Phi_\omega$, there is a child node $\theta$ of $\omega$ in $\mathcal T$ of level $r + 1$ with label corresponding to the recursive query.
    Assign to the edge $(\omega, \theta)$ the value
    \begin{align*}
        \kappa_{\omega, \theta} = c_B(A_{(v_\omega, u)}) \le c.
    \end{align*}
    
    Now we define the weight function $F$ on $\mathcal T$ as $F(\varrho) = 1$ and
    \begin{align*}
        \forall \omega \neq \varrho, \quad F(\omega) = \prod_e \kappa_e
    \end{align*}
    where the product takes over all edges in the unique path from $\varrho$ to $\omega$.
    Using the above bound for $\kappa_e$, we have $F(\omega) \le c^r$ for every node $\omega$ at level $r$.
    By \Cref{prop:spin-accuracy-recursion}, it holds that
    \begin{align*}
        \diam_H(\polytope_R(\Phi, v)) &\le \sum_{\omega \in \mathcal T~\text{at level}~R} F(\omega) \sum_{u~\text{is a neighbor of $v_\omega$ in $\Phi_\omega$}} \Delta_B(A_{(v_\omega, u)}) \\
        &+ 2\xi_P \sum_{r = 0}^{R - 1} \sum_{\omega \in \mathcal T~\text{at level}~r} F(\omega) \deg_{G_\omega}(v_\omega).
    \end{align*}

    On the other hand, nodes of $\mathcal T$ correspond one-to-one to self-avoiding walks from $v$ of length at most $R$.
    That is to say, a node $\omega \in \mathcal T$ at level $r$ corresponds to a unique self-avoiding walk from $v$ of length $r$.
    Therefore,
    \begin{align*}
        \sum_{\omega \in \mathcal T~\text{at level}~R} F(\omega) \sum_{u~\text{is a neighbor of $v_\omega$ in $\Phi_\omega$}} \Delta_B(A_{(v_\omega, u)}) &\le \overline\Delta \cdot c^R \abs{\SAW_G(v, R + 1)} \\
        &\le \overline\Delta \cdot c^R S_{R + 1}(G) \\
        &\le K n^\alpha D \overline\Delta (D \cdot c)^R \\
        &= K n^\alpha D \overline\Delta \rho^R
    \end{align*}
    and
    \begin{align*}
        2\xi_P \sum_{r = 0}^{R - 1} \sum_{\omega \in \mathcal T~\text{at level}~r} F(\omega) \deg_{G_\omega}(v_\omega) &\le 2\xi_P \sum_{r = 0}^{R - 1} c^r \abs{\SAW_G(v, r + 1)} \\
        &\le 2\xi_P \sum_{r = 0}^{R - 1} c^r \cdot Kn^\alpha D^{r + 1} \\
        &\le 2Kn^\alpha D \xi_P \sum_{r = 0}^{\infty} (D \cdot c)^r \\
        &= K n^\alpha D \cdot \frac{2\xi_P}{1 - \rho}.
    \end{align*}
    Combining all things together, we conclude \eqref{eq:spin-accuracy-bound}.

    For the second argument, substitute the choice of $R$ and $P$ in \eqref{eq:spin-accuracy-bound} and we obtain it.
\end{proof}

A polynomial-prefactor spatial-mixing statement can be vacuous on a graph family of logarithmic diameter.
Indeed, if every relevant distance satisfies $r\le a\ln n$ and a bound has the form $n^k\rho^r$ with $k\ge a\ln(1/\rho)$, then its right-hand side is at least one at every attainable distance.
In contrast, \eqref{eq:spin-accuracy-bound} directly controls the Hilbert diameter of the polytopes computed by the algorithm.
For the $O(\eps/n)$ accuracy required at each step of self-reduction, \eqref{eq:spin-depth-accuracy-parameters} gives $R,P=O_{K,\alpha,D,c}(\ln(n/\eps))$, including when $\alpha>0$.

\subsection{Approximation scheme to the partition function} \label{subsec:spin-approximation-scheme}
Now we are ready to show \Cref{thm:multi-spin-system}.
\begin{proof}[Proof of \Cref{thm:multi-spin-system}]
    For a fixed assignment $\sigma : V \to [q]$, pin vertices in a fixed order $v_1, \ldots, v_n$.
    Let $\Phi^{(j)}$ be the instance on $G^{(j)} = G[\{v_j,\ldots,v_n\}]$ obtained by pinning $v_k$ to $\sigma(v_k)$ for $k<j$, with the original interactions on the remaining edges and external fields
    \begin{align*}
        \lambda_u^{(j)}(a) = \lambda_u(a) \prod_{\substack{k<j\\(v_k,u)\in E}} A_{(v_k,u)}(\sigma(v_k),a), \qquad u\in V(G^{(j)}),\ a\in[q].
    \end{align*}
    Write $p_j = \mu_{\Phi^{(j)}}(v_j\gets\sigma(v_j))$.
    Since $G^{(j)}$ is a subgraph of $G$, $S_r(G^{(j)})\le S_r(G)\le Kn^\alpha D^r$ for all $r\ge0$, with $n$ denoting the original number of vertices.
    Then
    \begin{align*}
        \Z_G(\boldsymbol\lambda, \boldsymbol A) = w(\sigma) \prod_{j = 1}^n p_j^{-1}.
    \end{align*}
    Set $\eps' = \eps/(32n)$ and choose $R\ge R_{\eps'}, P\ge P_{\eps'}$ using \eqref{eq:spin-depth-accuracy-parameters} with the original $n$ and the fixed bound $c$.
    These integers can be chosen using fixed rational upper bounds on the constants and a fixed positive rational lower bound on $\ln(1/\rho)$; thus no input-dependent Birkhoff coefficient or transcendental value needs to be computed, and $R,P=O_{K,\alpha,D,c}(\ln(n/\eps))$.
    For every $j = 1, \ldots, n$, construct the polytope $\polytope_R(\Phi^{(j)}, v_j)$.
    Pick $\widehat p_j$ to minimize the value of entry $\sigma(v_j)$ in $\polytope_R(\Phi^{(j)}, v_j)$.
    By \Cref{fact:spin-polytope-soundness,cor:spin-accuracy}, it holds that
    \begin{align*}
        \forall j = 1, \ldots, n, \quad \e^{-\eps/(32n)} \le \widehat p_j / p_j \le 1.
    \end{align*}
    Write $\widehat Z = w(\sigma) \prod_{j = 1}^n \widehat p_j^{-1}$.
    We know that $1 \le \widehat Z/\Z_G(\boldsymbol\lambda, \boldsymbol A) \le \e^{\eps/32} \le 1 + \eps$.

    Now we consider the bit complexity.
    Whenever a vertex is pinned or removed in \Cref{def:pinnings}, all of its incident edges are deleted. Thus each original edge contributes a multiplicative factor to an external field at most once along any sequence of outer pinnings and recursive calls. In particular, for every derived instance $\Phi_\omega$ on $G_\omega$ and every remaining vertex $u$,
    \begin{align*}
        \bit(\lambda_u^{(\omega)}) \le O\ab(\bit(\lambda_u)+\sum_{\substack{e\in E(G)\setminus E(G_\omega)\\u\in e}}\bit(A_e)).
    \end{align*}
    Here the bit length of a vector or matrix is the sum of the bit lengths of its entries. The bound follows by adding numerator and denominator lengths when multiplying rational factors. Summing over the remaining vertices shows $\abs{\Phi_\omega}=O(\abs\Phi)$, so the unrounded external fields retain polynomial encoding length throughout the computation.
    When $R \ge n$, it must hold that $\eps = \e^{-\Omega(n)}$.
    Then we exactly compute $\Z_G(\boldsymbol\lambda, \boldsymbol A)$ by brute force enumerating all assignments and the bit complexity is bounded by the desired number.
    Otherwise, note that $R = O_{K, \alpha, D, c}(\ln(n/\eps))$ and $P = O_{K, \alpha, D, c}(\ln(n/\eps))$.
    The number of polytopes is at most
    \begin{align*}
        O\ab(n \sum_{r = 0}^R q^{2r} S_r(G)) = O\ab(K n^{\alpha + 1} (R + 1)(q^2D)^R) \le (n/\eps)^{C_1(q, K, \alpha, D, c)}.
    \end{align*}
    By \Cref{thm:Charness-Cooper-transformation,fact:spin-polytope-soundness}, each polytope has encoding length $O(q^3 (\abs\Phi + n \ln(n/\eps)))$, and solving the two LFPs over a child polytope takes at most $O(q^{20} (\abs\Phi + n \ln(n/\eps))^3)$ bit operations.
    Charge these two solves to the corresponding child node of the computation tree; the node count above already includes all choices of a spin pair and a neighboring vertex. The final marginal minimizations at the roots and the rational arithmetic used to form the constraints satisfy the same polynomial bound.
    Therefore, the total bit complexity is at most
    \begin{align*}
        O\ab((n/\eps)^{C_1(q, K, \alpha, D, c)} \cdot q^{20} (\abs\Phi + n\ln(n/\eps))^3) \le \ab(\frac{n + \abs{\Phi} + \bit(\eps)}{\eps})^C
    \end{align*}
    for some constant $C > 0$ depending only on $q, K, \alpha, D$ and $c$.
\end{proof}

\begin{proof}[Proof of \Cref{cor:general-Potts-model}]
    Set $A_e=A$ for every $e\in E$, where $A(i,i)=\beta$ and $A(i,j)=1$ for $i\ne j$.
    Since $\beta>0$, the interactions are positive and $c_B(A)=|\beta-1|/(\beta+1)<1/D$. Choose any fixed $c\in(0,1/D)$ with $c_B(A)\le c$ and apply \Cref{thm:multi-spin-system}.
    In particular, when $\beta=1$ one may use any such positive $c$; alternatively, the partition function is exactly $\prod_{v\in V}\sum_{a\in[q]}\lambda_v(a)$.
\end{proof}

\section{Proper Colorings on Sparse Random Graphs} \label{sec:random-graph-coloring}
In this section, we prove \Cref{thm:random-graph-coloring}.
Fix $\eta \in (0,1)$.
Throughout this section, $d \ge d_0(\eta)$ and the integer
$q \ge (2+\eta)d$ are fixed, where $d_0(\eta)$ is supplied by
\Cref{lem:coloring-parameter-choice}.
All asymptotic statements are for $n \to \infty$ with $\eta,d,q$ fixed.

We first construct the marginal recursion on permissive blocks
and bound the Hilbert diameter of the resulting polytopes.
We then establish random graph estimates that control both
the accumulated error and the total cost of the recursion.
The degree thresholds and auxiliary parameters are introduced
where they enter the analysis.

\subsection{List-coloring instances and permissive blocks}
Pinning a vertex to a color removes that color from the lists of its unpinned neighbors.
The residual instance is therefore naturally a list-coloring instance, and further feasible pinnings preserve this form.
We use this closure under pinning together with the permissive-block construction of \cite{Yin14,YZ16}.

A list-coloring instance $\Phi$ consists of an underlying graph $G_\Phi$ and color lists $L_\Phi$ such that $L_\Phi(v) \subseteq [q]$ collects feasible colors on $v$.
An assignment $\sigma : V(G_\Phi) \to [q]$ is a proper list-coloring in $\Phi$ if and only if it is a proper $q$-coloring and for every vertex $v \in V(G_\Phi)$, $\sigma(v) \in L_\Phi(v)$.
Let $\Omega(\Phi)$ be the collection of all proper list-colorings in $\Phi$ and set $\Z(\Phi) = \abs{\Omega(\Phi)}$.
With abuse of notation, for a vertex subset $B \subseteq V(G_\Phi)$, we use $\Omega_B$ to denote all proper partial list-colorings on $B$\footnote{A proper partial list-coloring on $B$ here is not necessarily able to extend to a proper list-coloring in $\Phi$.}.

Throughout the construction, $G=(V,E)$ denotes the original graph.
Define
\begin{equation}\label{eq:coloring-exceptional-set}
    X \defeq \set{v \in V \mid \deg_G(v) \ge q-1}.
\end{equation}
This set is used to construct permissive blocks.
It remains fixed when vertices are removed or colors are deleted
from the lists.

Then we introduce the following marginal upper bound for $v \notin X$ which has been established in \cite{Yin14}.
\begin{lemma}[\cite{Yin14}] \label{lem:coloring-marginal-upper-bound}
    For a feasible list-coloring instance $\Phi$ and every $v$ such that $\abs{L_\Phi(v)} > \deg_{G_\Phi}(v)$, it holds that
    \begin{align*}
        \forall c \in L_\Phi(v), \quad \mu_\Phi(v \gets c) \le \frac{1}{\abs{L_\Phi(v)} - \deg_{G_\Phi}(v)}.
    \end{align*}
\end{lemma}
For brevity, set
$$
    \forall v \notin X, \quad s_v \defeq q - \deg_G(v),\; \kappa_v \defeq \frac{1}{q - \deg_G(v) - 1}
$$
where $G$ is the original graph.

Now we introduce the concept of \emph{permissive blocks} in \cite{Yin14,YZ16}.
Throughout this section, $\partial_E B$ denotes the edge-boundary of $B$ in the current graph $G_\Phi$, with each edge oriented as $(u_i,v_i)$ where $u_i \in B$ and $v_i \notin B$.
Given a list-coloring instance $\Phi$, we say a connected vertex subset $B \subseteq V(G_\Phi)$ is a \emph{permissive block} if, for every $(u_i, v_i) \in \partial_E B$, $\abs{L_\Phi(v_i)} > \deg_{G_\Phi}(v_i) + 1$.
A permissive block $B$ remains permissive when a feasible partial coloring is pinned outside $B$: every removed neighbor decreases the degree by one and deletes at most one color from a remaining vertex's list.

We state here the following property of permissive blocks.
\begin{proposition} \label{prop:coloring-marginal-lower-bound}
    For a feasible list-coloring instance $\Phi$ and a permissive block $B \subseteq V(G_\Phi)$, every partial proper list-coloring on $B$ extends to a proper list-coloring in $\Phi$.
    Furthermore, for every proper list-coloring $\sigma$ on $B$, it holds that
    \begin{align*}
        \mu_\Phi(B \gets \sigma) \ge b_B \defeq 2^{-\abs{\partial_E B}} q^{-\abs{B}}.
    \end{align*}
\end{proposition}
\begin{proof}
    The feasibility of a feasible partial coloring on $B$ directly comes from \cite[Lemma 3.4]{Yin14}.
    For the marginal lower bound, fix a feasible partial coloring $\sigma$ on $B$.
    Construct intermediate instances $\Psi_0, \ldots, \Psi_{\abs{\partial_E B}}$ as: $\Psi_0$ is the instance derived by removing $B$ from $G_\Phi$, $L_\Phi(v)$ from color lists for every $v \in B$ and keeping all remaining color lists unchanged; for every $j = 1, \ldots, \abs{\partial_E B}$, generate $\Psi_j$ from $\Psi_{j - 1}$ by removing $\sigma(u_j)$ from the color list on $v_j$.
    For any intermediate instance $\Psi$, suppose that we need to delete color $c$ from the list of $v$ to obtain the instance $\Psi'$.
    Since $B$ is a permissive block, by \Cref{lem:coloring-marginal-upper-bound},
    \begin{align*}
        \mu_\Psi(v \gets c) \le \frac12.
    \end{align*}
    Therefore,
    \begin{align*}
        \Z(\Psi') &= \Z(\Psi) - \abs{\set{\tau~\text{is a proper list-coloring in $\Psi$} \;:\; \tau(v) = c}} \\
        &= \Z(\Psi) - \mu_\Psi(v \gets c) \cdot \Z(\Psi) \\
        &\ge \Z(\Psi)/2.
    \end{align*}
    Observe that $\Z(\Phi) \le \abs{\Omega_B} \cdot \Z(\Psi_0)$ and $\Z(\Psi_{\abs{\partial_E B}}) = \Z(\Phi \;|\; B \gets \sigma)$.
    Therefore,
    \begin{align*}
        \Z(\Phi \;|\; B \gets \sigma) \ge 2^{-\abs{\partial_E B}} \cdot \Z(\Psi_0) \ge 2^{-\abs{\partial_E B}} q^{-\abs B} \Z(\Phi).
    \end{align*}
    Then we conclude the marginal lower bound.
\end{proof}

For a feasible list-coloring instance $\Phi$ and a permissive block $B$, fix $\sigma, \tau \in \Omega_B$ and $i = 1, \ldots, m \defeq \abs{\partial_E B}$.
Define the intermediate instance $\Phi_{i, \sigma, \tau}$ to be the instance derived from $\Phi$ by:
\begin{itemize}
    \item the underlying graph is the induced subgraph of $G_\Phi$ by $V(G_\Phi) \setminus B$;
    \item the color lists are obtained by deleting $\sigma(u_j)$ from $L_\Phi(v_j)$ for $j = 1, \ldots, i - 1$, deleting $\tau(u_j)$ from $L_\Phi(v_j)$ for $j = i + 1, \ldots, m$ and keeping unchanged for other lists.
\end{itemize}
These intermediate instances are feasible whenever $\Phi$ is feasible and $B$ is permissive.
Indeed, removing $B$ first leaves a feasible instance; subsequently deleting at most one color per boundary edge keeps the list-size minus degree gap at each affected vertex at least its value in $\Phi$, which is at least two.
By \Cref{lem:coloring-marginal-upper-bound}, each color deletion therefore retains at least half of the current proper list-colorings, as in the proof of \Cref{prop:coloring-marginal-lower-bound}.
Then we have the following recursion of probability ratios.
\begin{lemma}[\cite{Yin14}] \label{lem:coloring-ratio-recursion}
    For every $\sigma, \tau \in \Omega_B$, it holds that
    \begin{align} \label{eq:coloring-ratio-recursion}
        \frac{\mu_\Phi(B \gets \sigma)}{\mu_\Phi(B \gets \tau)} = \prod_{i = 1}^m \frac{1 - \mu_{\Phi_{i, \sigma, \tau}}(v_i \gets \sigma(u_i))}{1 - \mu_{\Phi_{i, \sigma, \tau}}(v_i \gets \tau(u_i))}.
    \end{align}
\end{lemma}
We mention here that by \Cref{lem:coloring-marginal-upper-bound}, it is clear that every denominator in the right side of \eqref{eq:coloring-ratio-recursion} is at least $1/2$.
For every list-coloring instance $\Phi$ obtained by pinning a feasible partial coloring of the original $q$-coloring instance on $G$, it holds that $\Phi$ is feasible and
\begin{align} \label{eq:coloring-slack}
    \abs{L_\Phi(v)} \ge q - (\deg_G(v) - \deg_{G_\Phi}(v)).
\end{align}
The same inequality is preserved by every transition $\Phi \mapsto \Phi_{i,\sigma,\tau}$ on a permissive block: removing $B$ deletes all its incident boundary edges and deletes at most one color from a surviving vertex's list per such edge.
Together with the preceding feasibility argument, this proves both assertions inductively for all recursively generated instances, even though a transition uses two different partial colorings.
Thus, by \Cref{lem:coloring-marginal-upper-bound} and \eqref{eq:coloring-slack}, for $v \in V(G_\Phi) \setminus X$, $\abs{L_\Phi(v)} - \deg_{G_\Phi}(v) \ge s_v \ge 2$ and the marginal upper bound is at most $1/s_v$.

\subsection{Rational polytope for list-coloring instances}
Denote by $\mathcal C(H, v)$ the maximal connected component containing $v$ in $H$.
For every $v \in V(G_\Phi)$, set
\begin{align} \label{eq:coloring-permissive-block}
    B_\Phi(v) = \begin{cases}
        \mathcal C(G[X], v) \cap V(G_\Phi), & v \in X \\
        \set{v} \cup \ab(\bigcup_{C \sim v} (C \cap V(G_\Phi))), & v \notin X
    \end{cases}
\end{align}
where $C \sim v$ enumerates all connected components in $G[X]$ containing a neighbor of $v$ in $G_\Phi$.
We restrict accessibility to the pinnings used by the counting algorithm.
Given a feasible $q$-coloring instance $\Phi_0$ on $G$, an initial instance is obtained by pinning a feasible partial coloring either on a union of entire connected components of $G[X]$, or on $X \cup U$ for some $U \subseteq V \setminus X$.
We call a list-coloring instance $\Phi$ \emph{accessible} if it is obtained from such an initial instance by repeatedly applying the transitions in \Cref{lem:coloring-ratio-recursion}, each time removing a block $B_\Phi(v)$ from \eqref{eq:coloring-permissive-block}.

Every accessible instance satisfies the following structural invariant: for each connected component $C$ of $G[X]$, either $C \subseteq V(G_\Phi)$ or $C \cap V(G_\Phi)=\emptyset$.
This holds for the initial pinnings and is preserved because each block in \eqref{eq:coloring-permissive-block} removes a union of entire surviving components of $G[X]$, together with at most one vertex outside $X$.
Since $G_\Phi$ is induced by its surviving vertices, a surviving component remains connected.
Consequently, $B_\Phi(v)$ is connected: it is either one such component, or a vertex $v \notin X$ together with the surviving components adjacent to $v$.
Moreover, every boundary edge $(u_i,v_i) \in \partial_E B_\Phi(v)$ has $v_i \notin X$, because any component of $G[X]$ adjacent to the block is already included in it.
By \eqref{eq:coloring-slack}, this implies $\abs{L_\Phi(v_i)}-\deg_{G_\Phi}(v_i)\ge s_{v_i}\ge2$, so these blocks are permissive and all recursive instances remain feasible.

An \emph{admissible query} $(\Phi,B,t)$ henceforth consists of an accessible instance $\Phi=(G_\Phi,L_\Phi)$, a block $B=B_\Phi(v)$ for some $v\in V(G_\Phi)$, and a nonnegative integer depth $t$.
We now construct its rational polytope $\polytope_t(\Phi,B)$.
\begin{definition}[Rational polytope for list colorings] \label{def:coloring-polytope}
    For an admissible query $(\Phi,B,t)$, the variables of $\polytope_t(\Phi, B)$ are $(z_\sigma)_{\sigma \in \Omega_B}$.
    When $\partial_E B = \emptyset$, $\polytope_t(\Phi, B)$ is the exact uniform polytope
    \begin{align*}
        \forall \sigma \in \Omega_B, \quad z_\sigma = \frac{1}{\abs{\Omega_B}}.
    \end{align*}
    
    When $\partial_E B \neq \emptyset$, list $\partial_E B$ in order $(u_1, v_1), \ldots, (u_m, v_m)$.
    The polytope $\polytope_t(\Phi, B)$ at first contains:
    \begin{itemize}
        \item \textbf{Marginal lower bound.}
        \begin{align} \label{eq:coloring-LFP-marginal-lower-bound}
            \forall \sigma \in \Omega_B, \quad z_\sigma \ge b_B = 2^{-m} q^{-\abs B}.
        \end{align}
        \item \textbf{Normalization.}
        \begin{align} \label{eq:coloring-LFP-normalization}
            \sum_{\sigma \in \Omega_B} z_\sigma = 1.
        \end{align}
        \item \textbf{Marginal upper bound.}
        \begin{align} \label{eq:coloring-LFP-marginal-upper-bound}
            \forall v \in B \setminus X, c \in L_\Phi(v), \quad \sum_{\sigma \in \Omega_B : \sigma(v) = c} z_\sigma \le \frac{1}{s_v}.
        \end{align}
    \end{itemize}
    
    We separate cases based on the value of the depth $t$.
    \begin{itemize}
        \item $t = 0$.
        For every pair of distinct partial colorings $\sigma, \tau \in \Omega_B$, the polytope contains
        \begin{align} \label{eq:coloring-LFP-depth-zero}
            2^{-m} z_\tau \le z_\sigma \le 2^m z_\tau.
        \end{align}

        \item $t > 0$.
        For every pair of distinct partial colorings $\sigma, \tau \in \Omega_B$ and $i = 1, \ldots, m$, construct an instance $\Phi_i = \Phi_{i, \sigma, \tau}$.
        Recursively obtain the polytope $\polytope_i = \polytope_{t - 1}(\Phi_i, B_{\Phi_i}(v_i))$.
        For a probability vector $y \in \polytope_i$, we define
        \begin{align*}
            p_c(y) \defeq \sum_{\pi : \pi(v_i) = c} y_\pi
        \end{align*}
        with $p_c(y) = 0$ if $c$ is absent.
        Solve the following two LFPs:
        \begin{align*}
            \ell_{i, \sigma, \tau}^* \defeq \min_{y \in \polytope_i}\; \frac{1 - p_{\sigma(u_i)}(y)}{1 - p_{\tau(u_i)}(y)}, \quad
            u_{i, \sigma, \tau}^* \defeq \max_{y \in \polytope_i}\; \frac{1 - p_{\sigma(u_i)}(y)}{1 - p_{\tau(u_i)}(y)},
        \end{align*}
        and set $\underline\ell_{i, \sigma, \tau} = \mathsf{round}_P^-(\ell_{i, \sigma, \tau}^*), \overline u_{i, \sigma, \tau} = \mathsf{round}_P^+(u_{i, \sigma, \tau}^*)$.
        Setting $L_{\sigma, \tau}^t = \prod_{i = 1}^m \underline\ell_{i, \sigma, \tau}$ and $U_{\sigma, \tau}^t = \prod_{i = 1}^m \overline u_{i, \sigma, \tau}$, the polytope then contains
        \begin{align} \label{eq:coloring-LFP-ratio-recursion}
            L_{\sigma, \tau}^t \cdot z_\tau \le z_\sigma \le U_{\sigma, \tau}^t \cdot z_\tau.
        \end{align}
        For identical partial colorings, use the convention $L_{\sigma,\sigma}^t=U_{\sigma,\sigma}^t=1$.
    \end{itemize}
\end{definition}

\begin{fact}[Well-definedness and soundness] \label{fact:coloring-polytope-soundness}
    At every depth the recursive construction is well-defined, the true marginal vector on $B$ belongs to $\polytope_t(\Phi, B)$ and the bit complexity to encode $\polytope_t(\Phi, B)$ is at most $O(q^{4 \abs B} \cdot (\abs B + \abs{\partial_E B} \cdot P))$.
\end{fact}
\begin{proof}
    The proof is similar to \Cref{fact:spin-polytope-soundness}.
    When the polytope contains only an exact point (the case $\partial_E B = \emptyset$), the fact is trivial.
    Otherwise, at depth $0$, it is direct to see the true marginal vector $p_B$ on $B$ satisfies constraints \eqref{eq:coloring-LFP-marginal-lower-bound}, \eqref{eq:coloring-LFP-normalization} and \eqref{eq:coloring-LFP-marginal-upper-bound}.
    By \Cref{lem:coloring-ratio-recursion}, it is easy to verify $p_B$ satisfies \eqref{eq:coloring-LFP-depth-zero} since each numerator and denominator lies in $[1/2,1]$.
    The bit complexity to encode $\polytope_0(\Phi, B)$ is at most $O(q^{4\abs B}(\abs B + \abs{\partial_E B}))$ since we have $q^{\abs B}$ variables and $O(q^{2\abs B})$ constraints with each constraint encoded with at most $O(q^{2\abs B} (\abs B + \abs{\partial_E B}))$ bits.

    Suppose that the assertion holds through depth $t - 1$.
    It is trivial to see the true marginal vector $p_B$ satisfies \eqref{eq:coloring-LFP-marginal-lower-bound}, \eqref{eq:coloring-LFP-normalization} and \eqref{eq:coloring-LFP-marginal-upper-bound}.
    By induction and \Cref{lem:coloring-ratio-recursion}, it is not hard to see $p_B$ also satisfies the constraint \eqref{eq:coloring-LFP-ratio-recursion} even after applying rounding oracles.
    Therefore we conclude that $p_B \in \polytope_t(\Phi, B)$.
    To compute the number of bits to encode $\polytope_t(\Phi, B)$, note that we have $q^{\abs B}$ variables and $O(q^{2\abs B})$ linear constraints.
    Every boundary vertex $v_i$ lies outside $X$ by the structural invariant above.
    Hence $p_c(y)\le 1/s_{v_i}\le 1/2$ for all $y\in\polytope_i$, by \eqref{eq:coloring-LFP-marginal-upper-bound}; if the child polytope is an exact singleton, the same bound follows from \Cref{lem:coloring-marginal-upper-bound}.
    This also proves that each LFP denominator is positive throughout its feasible polytope.
    For a constraint \eqref{eq:coloring-LFP-ratio-recursion}, note that for each $i = 1, \ldots, m$, it holds that
    \begin{align*}
        \frac12 \le \ell_{i, \sigma, \tau}^* \le u_{i, \sigma, \tau}^* \le 2.
    \end{align*}
    Hence, it takes $O(P)$ bits to encode the rounded rational number $\underline \ell_{i, \sigma, \tau}, \overline u_{i, \sigma, \tau}$ and thus $O(\abs{\partial_E B} \cdot P)$ bits to encode $L_{\sigma, \tau}^t$ and $U_{\sigma, \tau}^t$.
    Therefore, the bit complexity to encode $\polytope_t(\Phi, B)$ is at most $O(q^{4\abs B}(\abs B + \abs{\partial_E B} \cdot P))$.
\end{proof}

Recall that in \Cref{sec:multi-spin-connective-constant}, for an admissible rational polytope $\polytope$, we use the Hilbert distance to denote the diameter of $\polytope$:
\begin{align*}
    \diam_H(\polytope) = \max_{x, y \in \polytope} d_H(x, y).
\end{align*}
The following lemma for the contraction of Hilbert distance is the key ingredient which is implicitly stated in \cite{GK12,GKM15,Yin14}.
\begin{lemma} \label{lem:coloring-Hilbert-contraction}
    Let $p,p'$ be positive probability vectors on a common nonempty state space $\Omega_+$, with every coordinate at most $1/s$ for $s\ge2$.
    Extend both vectors by zero to any larger state space $\Omega\supseteq\Omega_+$.
    Then, for every $a,b\in\Omega$,
    \begin{align*}
        \abs{\ln\frac{1 - p_a}{1 - p'_a} - \ln\frac{1 - p_b}{1 - p'_b}} \le \frac{d_H(p, p')}{s - 1}.
    \end{align*}
    Here the Hilbert distance is taken only on $\Omega_+$, including when $a$ or $b$ is absent from that support.
\end{lemma}
A clear proof of \Cref{lem:coloring-Hilbert-contraction} is included in \Cref{sec:contraction-proof} for completeness.
For a nonexceptional child root of original degree $k$, the polytope enforces $p_a\le1/(q-k)$ for every projected color marginal.
Thus \Cref{lem:coloring-Hilbert-contraction} supplies the coefficient $1/(q-k-1)$.
The improvement over the coefficient $2/(q-k-1)$ in the low-degree proper-coloring case of \cite[Theorem~7]{YZ16} comes from controlling $(1-p_a)/(1-p_b)$ over one common feasible marginal vector rather than combining separate coordinate extremes.
With \Cref{lem:coloring-Hilbert-contraction}, the following recursion controls the accuracy of $\polytope_t(\Phi, B)$ by its child polytopes.
\begin{proposition} \label{prop:coloring-accuracy-recursion}
    Fix an admissible feasible query $(\Phi, B, t)$.
    When $\partial_E B = \emptyset$, $\diam_H(\polytope_t(\Phi, B)) = 0$.
    Otherwise, list $\partial_E B$ in a fixed order $(u_1, v_1), \ldots, (u_m, v_m)$.
    It holds that
    \begin{align*}
        \diam_H(\polytope_t(\Phi, B)) \le \begin{cases}
            2m\ln2, & t = 0 \\
            \max_{\sigma, \tau \in \Omega_B} \sum_{i = 1}^m \ab(\kappa_{v_i} \cdot \diam_H(\polytope_{t - 1}(\Phi_{i, \sigma, \tau}, B_{\Phi_{i, \sigma, \tau}}(v_i))) + 2\xi_P), & t > 0
        \end{cases}\;.
    \end{align*}
\end{proposition}
\begin{proof}
    When $\partial_E B = \emptyset$, the rational polytope $\polytope_t(\Phi, B)$ contains exactly one feasible point representing the uniform distribution on $\Omega_B$.
    Therefore, $\diam_H(\polytope_t(\Phi, B)) = 0$.

    Suppose that $\partial_E B = \set{(u_1, v_1), \ldots, (u_m, v_m)}$ for some $m \ge 1$.
    For $t = 0$, by \eqref{eq:coloring-LFP-depth-zero}, for every $\sigma, \tau \in \Omega_B$,
    \begin{align*}
        -m\ln2 \le \ln\frac{z_\sigma}{z_\tau} \le m\ln2.
    \end{align*}
    Thus $\diam_H(\polytope_0(\Phi, B)) \le 2m\ln2$.
    
    At positive depth $t > 0$, the ratio width is zero for $\sigma=\tau$ by convention, so fix distinct $\sigma, \tau \in \Omega_B$.
    For $i = 1, \ldots, m$, consider the child rational polytope $\polytope_i = \polytope_{t - 1}(\Phi_{i, \sigma, \tau}, B_{\Phi_{i, \sigma, \tau}}(v_i))$.
    Since $v_i \notin X$ by the structural invariant of admissible queries, for every $y \in \polytope_i$, by \eqref{eq:coloring-LFP-marginal-upper-bound} (or \Cref{lem:coloring-marginal-upper-bound} for an exact child polytope),
    \begin{align*}
        p_c(y) \le \frac{1}{s_{v_i}}.
    \end{align*}
    Write $B_i=B_{\Phi_{i,\sigma,\tau}}(v_i)$ and define its projected active support by
    \begin{align*}
        \Omega_{i,+}\defeq\set{c\in L_{\Phi_{i,\sigma,\tau}}(v_i)\mid \exists\pi\in\Omega_{B_i},\ \pi(v_i)=c}.
    \end{align*}
    This set is nonempty and independent of $y$.
    On this support, $p_c(y)\ge b_{B_i}>0$ by \eqref{eq:coloring-LFP-marginal-lower-bound}; the same conclusion holds for an exact child polytope since its coordinates equal $1/\abs{\Omega_{B_i}}\ge q^{-\abs{B_i}}$.
    Outside this support, $p_c(y)=0$ for every feasible $y$.
    We therefore take $d_H(p(y),p(y'))$ on $\Omega_{i,+}$ and apply the zero-extension version of the contraction lemma when either of the colors $\sigma(u_i),\tau(u_i)$ is absent.
    Hence by \Cref{lem:coloring-Hilbert-contraction}, for every $y, y' \in \polytope_i$,
    \begin{align*}
        \abs{\ln\frac{1 - p_{\sigma(u_i)}(y)}{1 - p_{\tau(u_i)}(y)} - \ln\frac{1 - p_{\sigma(u_i)}(y')}{1 - p_{\tau(u_i)}(y')} } \le \frac{d_H(p(y), p(y'))}{s_{v_i} - 1} \le \kappa_{v_i} d_H(y, y')
    \end{align*}
    where the last inequality follows because, for each active color $c$, $p_c(y)/p_c(y')$ is a weighted average of the ratios $y_\pi/y'_\pi$ over $\pi(v_i)=c$.
    Together with the rounding accuracy \eqref{eq:rounding-accuracy},
    \begin{align*}
        \ln\frac{U_{\sigma, \tau}^t}{L_{\sigma, \tau}^t} \le \sum_{i = 1}^m \ab(\kappa_{v_i} \diam_H(\polytope_i) + 2\xi_P).
    \end{align*}
    Thus a similar argument to the proof of \Cref{prop:spin-accuracy-recursion} shows that
    \begin{align*}
        \diam_H(\polytope_t(\Phi, B)) \le \max_{\sigma, \tau \in \Omega_B} \ln\frac{U_{\sigma, \tau}^t}{L_{\sigma, \tau}^t},
    \end{align*}
    which concludes the proposition.
\end{proof}

\subsection{Random graph properties}\label{sec:random-graph-properties}
The random graph estimates below serve two purposes.
Decay of weighted self-avoiding walks controls the accumulated
truncation error, while bounds on exceptional components and
weighted block walks control the block sizes and the total
recursive cost.

To choose a degree cutoff $D$, we balance the rarity of vertices
above the cutoff against contraction below it.
For a vertex of degree at most $D$, the contraction coefficient
$\kappa_v$ is at most $1/(q-D-1)$.
Ignoring integer rounding and additive constants, the two
requirements are
\[
    D>d
    \qquad\text{and}\qquad
    \frac{d}{q-D}<1,
\]
or equivalently $d<D<q-d$.
This balance explains the role of $q>2d$ in the present analysis.
We choose
\begin{equation}\label{eq:coloring-degree-parameters}
    D \defeq \left\lfloor\frac{d+q}{3}\right\rfloor,
    \qquad
    \gamma \defeq \frac{1}{q-D-1},
\end{equation}
and define the exceptional set for the probabilistic estimates by
\begin{equation}\label{eq:coloring-high-degree-set}
    H_D \defeq \set{v \in V \mid \deg_G(v)>D}.
\end{equation}

For the degree tail estimates, write
\[
    h(x) \defeq (1+x)\ln(1+x)-x
    \qquad (x>0),
\]
and set
\begin{equation}\label{eq:coloring-tail-parameters}
\begin{aligned}
    \delta &\defeq \frac{q/d-2}{6},
    & r &\defeq \e^{-dh(\delta)},\\
    \lambda_0 &\defeq \frac{3+\eta/2}{3+\eta},
    & \rho &\defeq \frac{1+\lambda_0}{2}.
\end{aligned}
\end{equation}
Here $r$ bounds the conditional probability that a vertex is
exceptional when at most three incident edges have been prescribed.
The constants $\lambda_0$ and $\rho$ are the decay rates used
in the expectation and high-probability bounds, respectively.

By \Cref{lem:coloring-parameter-choice}, these parameters satisfy
\begin{align}
    D-3 &> (1+\delta)d,
    \qquad D<q-1,
    \label{eq:coloring-threshold-estimates}\\
    d(\gamma+\sqrt r) &\le \lambda_0<\rho<1,
    \label{eq:coloring-path-parameter-estimate}\\
    \e d r^{1/4} &\le \frac14.
    \label{eq:coloring-component-parameter-estimate}
\end{align}
The first estimate leaves room for up to three prescribed
incident edges in the degree tail bounds.
The second yields decay of path weights, and the third controls
the sizes of exceptional components.
In particular, $X \subseteq H_D$.
The set $X$ determines the permissive blocks, whereas $H_D$ is
used to bound their sizes and to control the propagation of error.

\subsubsection{Path weights}
For a graph $G$, we define the weight function $w : V(G) \to \mathbb R$ as
\begin{align*}
    \forall v \in V(G), \quad w(v) = \begin{cases}
        \gamma & \deg_G(v) \le D \\
        1 & \deg_G(v) > D
    \end{cases}\;.
\end{align*}
For $\ell \ge 1$, define
\begin{align*}
    W_\ell(G) \defeq \sum_{v_0, \ldots, v_\ell} \prod_{i = 1}^\ell w(v_i)
\end{align*}
where the sum takes over all simple paths of length $\ell$ in $G$.
We have the following exponential decay of $W_\ell(G)$ with high probability over the random graph $G$.
\begin{lemma} \label{lem:path-weight-concentration}
    For every positive integer $n \ge \max\{d,2 + 12/\eta\}$, it holds that
    \begin{align*}
        \Pr[G \sim \mathcal G(n, d/n)]{\exists 1 \le \ell < n : W_\ell(G) > n^6 \rho^\ell} \le n^{-4}.
    \end{align*}
\end{lemma}
\begin{proof}
    Fix $\ell \ge 1$ and a simple path $P = v_0, \ldots, v_\ell$ of length $\ell$.
    Conditional on the event $\mathcal E_P = \set{P~\text{is a simple path in}~G}$, for every $v \in V$, write $\deg_G(v) = f_v + X_v$ where $f_v = 1$ if $v = v_0$ or $v = v_\ell$, $f_v = 2$ if $v \in \set{v_1, \ldots, v_{\ell - 1}}$ and $f_v = 0$ otherwise.
    Note that $X_v$ is a random variable dominated by a binomial random variable $Y \sim \mathrm{Bin}(n, d/n)$.
    For the event $\deg_G(v) > D$, it must hold that
    \begin{align*}
        X_v = \deg_G(v) - f_v > D - 2 > (1 + \delta)d
    \end{align*}
    where the last inequality follows from
    \eqref{eq:coloring-threshold-estimates}.
    Therefore by the Chernoff bound,
    \begin{align*}
        \Pr{v \in H_D \mid \mathcal E_P} \le \Pr{Y \ge (1 + \delta)d} \le \e^{-d h(\delta)} = r.
    \end{align*}

    Now we turn to the random variables $\set{\id{v \in H_D}}_{v \in V}$.
    Conditional on $\mathcal E_P$, all edge indicators outside the prescribed path remain independent.
    For any fixed vertex set $S$, the event $\{v\in H_D\}$ is a function of the remaining indicators incident to $v$, and each remaining edge indicator is read by at most the two endpoint events.
    Finner's generalized H\"{o}lder inequality for this conditional read-two family \cite{Finner92} therefore gives
    \begin{align*}
        \E{\prod_{v \in S} \id{v \in H_D} \mid \mathcal E_P} \le \prod_{v \in S} \E{\id{v \in H_D}^2 \mid \mathcal E_P}^{1/2} \le r^{\abs S/2}.
    \end{align*}
    Since $w(v) \le \gamma + \id{v \in H_D}$,
    \begin{align*}
        \E{\prod_{i = 1}^\ell w(v_i) \mid \mathcal E_P} &\le \sum_{S \subseteq [\ell]} \gamma^{\ell - \abs{S}} \Pr{\forall i \in S, v_i \in H_D \mid \mathcal E_P} \\
        &\le \sum_{i = 0}^\ell \binom{\ell}{i} \gamma^{\ell - i} r^{i/2} \\
        &\le (\gamma + \sqrt r)^\ell.
    \end{align*}
    By the law of total expectation and
    \eqref{eq:coloring-path-parameter-estimate},
    \begin{align*}
        \E{W_\ell(G)} \le nd^{\ell}(\gamma + \sqrt{r})^\ell \le n\lambda_0^\ell.
    \end{align*}
    Then by the Markov inequality and the union bound,
    \begin{align*}
        \Pr{\exists 1 \le \ell < n, W_\ell(G) > n^6 \rho^\ell} \le n \cdot n^{-6} \frac{\lambda_0/\rho}{1 - \lambda_0/\rho} \le n^{-4}
    \end{align*}
    for every $n \ge \frac{\lambda_0}{\rho - \lambda_0} = 2 + 12/\eta$.
\end{proof}

\subsubsection{Exceptional component size and incident volume}
For the induced subgraph $G[H_D]$, we first consider the size of each component in it.
\begin{lemma} \label{lem:exceptional-component-size}
    With probability at least $1 - \frac12 n^{-4}$ over $G \sim \mathcal G(n, d/n)$, every connected component in $G[H_D]$ has size at most $4\ln n$.
\end{lemma}
\begin{proof}
    For a component $C$ of $G[H_D]$ of size $t$, take a spanning tree $T$ of $C$.
    Conditional on $T$ being present in $G$, all other edge indicators remain independent, and each is read by at most the degree events of its two endpoints.
    Thus at least half of $V(T)$ has tree degree at most $3$.
    For each of these vertices, \eqref{eq:coloring-threshold-estimates}
    gives the same conditional degree tail bound $r$.
    Applying the conditional read-two inequality as in the proof of
    \Cref{lem:path-weight-concentration} gives
    \begin{align*}
        \Pr{T \subseteq H_D \mid T~\text{in}~G} \le r^{t/4}.
    \end{align*}
    A component of size at least $t$ contains a tree on exactly $t$ vertices.
    Thus, by the union bound and Cayley's formula,
    \begin{align*}
        \Pr{\exists \text{a component $C$ in $H_D$ of size $\ge t$}} &\le \binom{n}{t} t^{t - 2} \ab(\frac{d}{n})^{t - 1} r^{t/4} \\
        &\le \frac{\e^t n^t}{t^t} t^{t - 2} \cdot \frac{d^{t - 1}}{n^{t - 1}} \cdot r^{t/4} \\
        &\le \frac{n}{d} \ab(\e d r^{1/4})^t.
    \end{align*}
    Using \eqref{eq:coloring-component-parameter-estimate} and
    taking $t=\lceil 4\ln n\rceil$ gives the claimed failure probability.
\end{proof}

The second result is about the incident volume of a component in $G[H_D]$.
\begin{lemma} \label{lem:exceptional-incident-volume}
    For all sufficiently large $n$, with probability at least $1 - \frac{1}{2}n^{-4}$, every component $C$ in $G[H_D]$ with $\abs C \le 4\ln n$ satisfies
    \begin{align*}
        \sum_{v \in C} \deg_G(v) \le C_{\mathrm{vol}} \ln n
    \end{align*}
    where $C_d = 2 + d(\e - 1) + d(\e^2 - 1)/4$ and $C_{\mathrm{vol}} = 4C_d + 4\ln d + 14$.
\end{lemma}
\begin{proof}
    Fix a vertex subset $S$ of size $t$ and let $T$ be a tree on $S$.
    Define a random variable $Y_S \defeq \sum_{v \in S} \deg_G(v)$.
    Conditional on the event $\mathcal E_T$ that $T$ is in $G$, it holds that
    \begin{align*}
        Y_S = 2(t - 1) + \sum_{e \in S \times (V \setminus S)} X_e + 2 \sum_{e \in \binom{S}{2} \setminus E(T)} X_e
    \end{align*}
    where the indicators $X_e = \id{e \in E(G)}$ for edges $e \notin E(T)$ are independent identically distributed Bernoulli random variables with parameter $p = d/n$, conditional on $\mathcal E_T$.
    For every $t \le n/2$,
    \begin{align*}
        \E{\e^{Y_S} \mid \mathcal E_T} &= \e^{2(t - 1)} \cdot (1 + p(\e - 1))^{t(n - t)} \cdot (1 + p(\e^2 - 1))^{\binom{t}{2} - (t - 1)} \\
        &\le \exp\ab(2(t - 1) + p(\e - 1)t(n - t) + p(\e^2 - 1)\ab(\binom{t}{2} - (t - 1))) \\
        &\le \exp\ab(2t + t \cdot d(\e - 1) + t \cdot d(\e^2 - 1)/4) = \e^{C_d \cdot t}.
    \end{align*}
    By Markov's inequality and the union bound, for each $t \le 4\ln n$,
    \begin{align*}
        \Pr{\exists \text{a component $C$ in $G[H_D]$ of size $t$}, Y_C \ge C_{\mathrm{vol}}\ln n} \le \binom{n}{t} t^{t - 2} \ab(\frac{d}{n})^{t - 1} \e^{C_d t - C_{\mathrm{vol}} \ln n}.
    \end{align*}
    The crude inequality $\binom{n}{t} \le (\e n/t)^t$ gives that this failure probability is at most
    \begin{align*}
        \frac{\e^t n^t}{t^t} t^{t - 2} \frac{d^{t - 1}}{n^{t - 1}} \e^{C_d t - C_{\mathrm{vol}} \ln n} \le \frac{n}{d} \e^{-C_{\mathrm{vol}} \ln n} \ab(d \e^{C_d + 1})^t.
    \end{align*}
    Therefore,
    \begin{align*}
        \Pr{\exists \text{a component $C$ in $G[H_D]$},\ \abs C\le4\ln n,\ Y_C \ge C_{\mathrm{vol}}\ln n} &\le \frac{n}{d} \cdot \e^{-C_{\mathrm{vol}} \ln n} \sum_{t = 1}^{\lfloor4\ln n\rfloor} \ab(d \e^{C_d + 1})^t \\
        &\le \frac{n}{d} \cdot \e^{-C_{\mathrm{vol}} \ln n} (4\ln n) \ab(d \e^{C_d + 1})^{4\ln n} \\
        &\le n^{-5}.
    \end{align*}
    Thus we conclude the lemma.
\end{proof}

\subsubsection{Permissive block walks}
We next bound the total weight of permissive block walks.
The error analysis in \Cref{lem:coloring-accuracy-bound} follows one maximizing pair of block colorings at each node.
To bound running time, however, we must sum over every pair of block colorings and every recorded transition edge.
The fixed exponent $K$ below absorbs these choices and the local optimization cost associated with each block; see \eqref{eq:coloring-block-cost}.
Set
\begin{equation}\label{eq:coloring-block-parameters}
    K \defeq 200,\qquad
    J \defeq (q-1)\ln2-d,\qquad
    y \defeq 4dq^K\e^{-J/2}.
\end{equation}
By \Cref{lem:coloring-parameter-choice},
\begin{equation}\label{eq:coloring-block-parameter-estimate}
    \e y \le \frac14.
\end{equation}
This estimate controls the sum over additional vertices
in the witness enumeration below.

For an accessible instance $\Phi$ on $G$, a permissive block walk of length $\ell$ is a sequence of permissive blocks generated by the following process.
Choose a root vertex $v_0 \notin X$ and set $B_0 = B_{\Phi_0}(v_0)$ where $\Phi_0 = \Phi$.
At step $j \ge 0$ to generate $B_{j + 1}$, let $\Phi_{j + 1}$ be the instance generated by removing $B_j$ from $\Phi_j$.
Choose an edge $(u, v) \in \partial_E B_j$, set the next root $v_{j + 1} = v$ and let $B_{j + 1} = B_{\Phi_{j + 1}}(v_{j + 1})$.
Thus the blocks are disjoint and we record the blocks, roots and transition edges for a block walk.
The weight of this block walk is defined by $q^{K \sum_{j = 0}^\ell \abs{B_j}}$.
Let $\mathcal W_\ell(G)$ be the sum of these weights over all distinct records of blocks, roots and transition edges that can arise from any accessible starting instance on $G$.
Each record is counted once, regardless of its color lists or the sequence of pinnings that produced it.
Thus $\mathcal W_\ell(G)$ depends only on the original graph $G$ and the fixed set $X$, and bounds the corresponding sum for every accessible instance simultaneously.
\begin{lemma} \label{lem:block-walk-weight}
    For any sufficiently large positive integer $n$, with probability at least $1 - n^{-4}$ over $G \sim \mathcal G(n, d/n)$, for every $0 \le \ell < n$,
    \begin{align*}
        \mathcal W_\ell(G) \le n^7 \ab(4dq^K \e^{2y})^\ell.
    \end{align*}
\end{lemma}
\begin{proof}
    We first define a witness of a permissive block walk.
    For a permissive block walk $B_0, \ldots, B_\ell$ of length $\ell$ with roots $v_0, \ldots, v_\ell$,
    assume that the union of it contains $t$ vertices in $X$ and set $t_j = \abs{B_j \cap X}$.
    Its intra-block forest is defined as follows: for every $j = 0, \ldots, \ell$,
    \begin{itemize}
        \item pick a spanning tree in every connected component of $B_j \cap X$;
        \item for every spanning tree in $B_j \cap X$, pick an edge connecting $v_j$ and it.
    \end{itemize}
    This forest has exactly $t$ edges, and its component rooted at $v_j$ spans $B_j$.
    The full witness $\mathsf W$ consists of this forest together with the $\ell$ recorded transition edges $(u_j,v_{j+1})$, $0\le j<\ell$.
    Since these edges join consecutive disjoint blocks, $\mathsf W$ is a tree on $\bigcup_{j=0}^\ell B_j$, with exactly $\ell+t$ edges.
    Fix canonical choices of the intra-block trees and their root edges under the vertex order.
    The resulting witness, together with the ordered roots and recorded transitions, determines the block record: deleting the transition edges recovers each rooted block.
    In particular, the enumeration below bounds all distinct records across accessible instances without counting their lists or pinnings.

    For every $j < \ell$, there are at most $t_j + 1$ possible transitions between $B_j$ and $v_{j + 1}$.
    Hence there are at most $\prod_{j = 0}^{\ell - 1} (t_j + 1) \le 2^t$ possible recorded transitions.
    Meanwhile, by the rooted-forest form of Cayley's formula, for a fixed collection containing $\ell + 1$ roots and additional $t$ vertices in blocks, the number of intra-block forests is at most
    \begin{align*}
        (\ell + 1)(\ell + 1 + t)^{t - 1}.
    \end{align*}
    Then for every $t = 0, \ldots, n$, the total number of witnesses with recorded transitions containing $\ell + 1$ roots and additional $t$ vertices in blocks is at most
    \begin{align*}
        n(n - 1) \ldots (n - \ell) \binom{n - \ell - 1}{t} (\ell + 1)(\ell + 1 + t)^{t - 1} 2^t \le \frac{n^{\ell + 1 + t}}{t!} (\ell + 1)(\ell + 1 + t)^{t - 1} 2^t.
    \end{align*}

    Fix a full witness $\mathsf W$ with recorded transitions, containing $\ell + 1$ roots and additional $t$ vertices in blocks.
    Its $\ell+t$ distinct edges give $\Pr{\mathcal E_{\mathsf W}}=(d/n)^{\ell+t}$, where $\mathcal E_{\mathsf W}$ is the event that all these edges occur in $G$.
    For a vertex $v \in \mathsf W$ such that it is prescribed to be in $X$, set $f_v = \deg_{\mathsf W}(v)$.
    Conditional on $\mathcal E_{\mathsf W}$, it holds that
    \begin{align*}
        \deg_G(v) = f_v + Y_v, \quad \text{where}~Y_v \sim \mathrm{Bin}(n - 1 - f_v, d/n).
    \end{align*}
    By the Markov inequality,
    \begin{align*}
        \Pr{v \in X \mid \mathcal E_{\mathsf W}} \le 2^{-(q - 1 - f_v)} \E{2^{Y_v}} \le \e^{-J + f_v\ln 2}.
    \end{align*}
    Conditional on $\mathcal E_{\mathsf W}$, the unspecified edge indicators remain independent.
    Every such indicator is read by at most the two prescribed degree events at its endpoints, so the conditional read-two form of Finner's inequality used in \Cref{lem:path-weight-concentration} gives
    \begin{align*}
        \Pr{\text{All $t$ prescribed vertices are in $X$} \mid \mathcal E_{\mathsf W}} &\le \prod_v \Pr{v \in X \mid \mathcal E_{\mathsf W}}^{1/2} \\
        &\le \exp\ab(-\frac{Jt}{2} + \frac{\ln2}{2} \sum_v f_v) \\
        &\le \e^{-Jt/2} 2^{\ell + t}.
    \end{align*}
    since $\sum_v f_v \le 2(\ell + t)$.
    Therefore,
    \begin{align*}
        \E{\mathcal W_\ell(G)} &\le \sum_{t \ge 0} \frac{n^{\ell + 1 + t}}{t!} (\ell + 1)(\ell + 1 + t)^{t - 1} 2^t \frac{d^{\ell + t}}{n^{\ell + t}} q^{K(\ell + 1 + t)} \e^{-Jt/2} 2^{\ell + t} \\
        &= nq^K\ab(2dq^K)^\ell \sum_{t \ge 0} \frac{(\ell + 1)(\ell + 1 + t)^{t - 1}}{t!} \ab(4dq^K \e^{-J/2})^t.
    \end{align*}
    To compute the series, we consider the exponential rooted-tree generating function $T$:
    \begin{align*}
        T(x) = \sum_{m \ge 1} m^{m - 1} \frac{x^m}{m!}.
    \end{align*}
    For $0<x<1/\e$, this series satisfies $T(x) = x\e^{T(x)}$ and $T(x) \in (0, 1)$.
    Apply the following form of Lagrange inversion (see \cite{Gessel16,SW23} for more details)
    \begin{align*}
        [x^t] F(u(x)) = \frac{1}{t} [z^{t - 1}] F'(z) \phi(z)^t,\qquad t\ge1,
    \end{align*}
    for $u = x \phi(u)$ with choice $\phi(u) = \e^u, F(u) = \e^{(\ell + 1)u}$ and $u = T(x)$, and we obtain that
    \begin{align*}
        \sum_{t \ge 0} \frac{(\ell + 1)(\ell + 1 + t)^{t - 1}}{t!} x^t = \e^{(\ell + 1)T(x)}.
    \end{align*}
    Hence,
    \begin{align*}
        \E{\mathcal W_\ell(G)} \le nq^K\ab(2dq^K)^\ell \e^{(\ell + 1) T(y)}.
    \end{align*}
    Consider the fixed-point iteration
    \begin{align*}
        T_0 = 0, \quad \forall k \ge 1,\; T_k = y \e^{T_{k - 1}}.
    \end{align*}
    Since $T_1 = y \ge T_0$ and $x \mapsto y \e^x$ is increasing, by induction it holds that $T_0 \le T_1 \le T_2 \le \cdots$.
    For all $0 \le x \le 2y \le (1/2\e)$, it holds that $y\e^x \le y\e^{2y} \le 2y$.
    Therefore $\lim_{k \to \infty} T_k = T(y)$ and $0 \le T_k \le 2y$ for every $k \ge 0$.
    This means $T(y) \le 2y$.
    Then by the Markov inequality at the threshold $n^7\ab(4dq^K\e^{2y})^\ell$, as long as $n \ge q^K \e^{2y}$, it holds that
    \begin{align*}
        \Pr{\exists \ell < n, \mathcal W_\ell(G) \ge n^7 \ab(4dq^K\e^{2y})^\ell} \le \sum_{\ell = 0}^{n - 1} n^2 n^{-7} \cdot 2^{-\ell} \le n^{-4}
    \end{align*}
    concluding the lemma.
\end{proof}

\subsection{Approximation scheme to proper colorings}
To design our algorithm, we show that with high probability, the rational polytope $\polytope_t(\Phi, B)$ can induce a good approximation to the real marginal vector $\mu_\Phi(B \gets \cdot)$.

\begin{lemma} \label{lem:coloring-accuracy-bound}
    On the event of \Cref{lem:path-weight-concentration,lem:exceptional-component-size,lem:exceptional-incident-volume}, for every list-coloring instance $\Phi$ induced by pinning components in $G[X]$ from the original $q$-coloring instance, an unassigned connected component $B$ in $X$ and every $R = 0, \ldots, n$, it holds that
    \begin{align*}
        \diam_H(\polytope_R(\Phi, B)) \le \frac{C_{\mathrm{leaf}} n^6 \ln n}{1 - \rho} \rho^R + \frac{2(q - 1)n^6}{1 - \rho} \xi_P
    \end{align*}
    where $C_{\mathrm{leaf}} = 2\ln2 \cdot (q - 2) \cdot (1 + C_{\mathrm{vol}})$.
    Moreover, if $\diam_H(\polytope_R(\Phi, B)) \le \eps$, then for every $z \in \polytope_R(\Phi, B)$ and $\sigma \in \Omega_B$, $\e^{-\eps} \mu_\Phi(B \gets \sigma) \le z_\sigma \le \e^{\eps} \mu_\Phi(B \gets \sigma)$.
\end{lemma}
\begin{remark}\label{rem:coloring-residual-accuracy}
    We remark here that after $X$ is pinned, the residual instance $\Phi$ satisfies that $\abs{L_{\Phi}(v)} - \deg_{G_{\Phi}}(v) \ge q - \deg_G(v) \ge 2$.
    The same accuracy bound in \Cref{lem:coloring-accuracy-bound} holds for the singleton query $\polytope_R(\Phi',\{v\})$ in every instance $\Phi'$ obtained by a further feasible partial pinning in $\Phi$.
\end{remark}
\begin{proof}
    We only need to consider the non-trivial case.
    By \Cref{prop:coloring-accuracy-recursion}, at depth $0$,
    \begin{align*}
        \diam_H(\polytope_0(\Phi, B)) \le 2m\ln2.
    \end{align*}
    Every queried block is either a connected component in $G[X]$, or a vertex $v\notin X$ together with at most $q-2$ adjacent components of $G[X]$ that have not been removed.
    Each component of $G[X]$ is contained in a component of $G[H_D]$.
    Therefore, by \Cref{lem:exceptional-incident-volume},
    \begin{align*}
        m \le (q - 2) + (q - 2) C_{\mathrm{vol}} \ln n \le (q - 2) \, (1 + C_{\mathrm{vol}}) \ln n.
    \end{align*}
    Thus $\diam_H(\polytope_0(\Phi, B)) \le C_{\mathrm{leaf}} \ln n$, which proves the asserted diameter bound when $R=0$.
    In the remainder of the diameter estimate, assume $R>0$.

    For $R>0$, unfold the maximum in \Cref{prop:coloring-accuracy-recursion} by choosing a single maximizing pair $(\sigma,\tau)$ at each node, breaking ties by a fixed order.
    This gives a tree with one child for each boundary edge; branches ending at a block with empty boundary contribute zero diameter.
    We encode every remaining branch, including truncated branches, by a self-avoiding walk.
    Consider a branch of length $r$, with $1\le r\le R$,
    \begin{align*}
        B_0 \to B_1 \to \ldots \to B_r.
    \end{align*}
    If $B_0\subseteq X$, let $x_0$ be its smallest vertex; otherwise let $x_0$ be its unique vertex outside $X$.
    For $i=1,\ldots,r$, write $(u_{i-1},x_i)\in\partial_E B_{i-1}$ for the transition edge, so $x_i\notin X$ and $B_i=B_{\Phi_i}(x_i)$.
    For $i=0,\ldots,r-1$, choose the lexicographically first shortest path from $x_i$ to $u_i$ inside $G_{\Phi_i}[B_i]$.
    Concatenate these paths with the transition edges and stop at $x_r$, without adding a path inside $B_r$.
    Since the blocks are connected and pairwise disjoint, the result is a self-avoiding walk $P=(x_0=v_0,\ldots,v_\ell=x_r)$ of length $\ell\ge r$.
    All vertices of $P$ other than the roots belong to $X\subseteq H_D$, and therefore have weight $1$.
    Moreover, $\kappa_{x_i}\le w(x_i)$: for $\deg_G(x_i)\le D$ this follows from $\kappa_{x_i}=1/(q-\deg_G(x_i)-1)\le\gamma$, and otherwise from $x_i\notin X$ and $\kappa_{x_i}\le1$.
    Hence
    \begin{align*}
        \prod_{i = 1}^r \kappa_{x_i} \le \prod_{j = 1}^\ell w(v_j).
    \end{align*}
    This encoding is injective even when branches of different lengths are considered together.
    Starting from the known query at $B_0$, the first edge of $P$ leaving $B_0$ determines its child uniquely, because the maximizing pair at that node has been fixed.
    Repeat this reconstruction at the child block.
    The walk never returns to a removed block, and stops at the root of its final block, so both the transitions and the stopping depth are recovered.

    The nonzero leaf contributions come from branches of length $R$, and the preceding weight comparison bounds their sum by $C_{\mathrm{leaf}}\ln n\sum_{\ell=R}^{n-1}W_\ell(G)$.
    For the rounding terms, a boundary edge $(u_r,x_{r+1})$ at depth $0\le r<R$ contributes
    $2\xi_P\prod_{i=1}^r\kappa_{x_i}$.
    Encode the branch extended by that edge, stopping at $x_{r+1}$, by the same construction.
    Its length is at least $r+1\ge1$, and
    \begin{align*}
        \prod_{i=1}^r\kappa_{x_i}
        \le \frac{1}{w(x_{r+1})}\prod_{j=1}^\ell w(v_j)
        \le (q-1)\prod_{j=1}^\ell w(v_j),
    \end{align*}
    since $w(x_{r+1})\ge\gamma=1/(q-D-1)\ge1/(q-1)$.
    Injectivity across all depths bounds the total rounding contribution by $2(q-1)\xi_P\sum_{\ell=1}^{n-1}W_\ell(G)$.
    Thus \Cref{lem:path-weight-concentration} gives
    \begin{align*}
        \diam_H(\polytope_R(\Phi, B)) &\le C_{\mathrm{leaf}} \ln n \sum_{\ell = R}^{n - 1} W_\ell(G) + 2(q - 1)\xi_P \sum_{\ell = 1}^{n - 1} W_\ell(G) \\
        &\le C_{\mathrm{leaf}} n^6 \ln n \sum_{\ell \ge R} \rho^\ell + 2(q - 1)n^6 \cdot \xi_P \sum_{\ell \ge 1} \rho^\ell \\
        &\le \frac{C_{\mathrm{leaf}} n^6 \ln n}{1 - \rho} \rho^R + \frac{2(q - 1)n^6}{1 - \rho} \xi_P.
    \end{align*}

    For the approximation to $p_B(\cdot) = \mu_\Phi(B \gets \cdot)$, note that $d_H(z, p_B) \le \eps$.
    Normalization implies $\min_\sigma z_\sigma/p_B(\sigma)\le1\le\max_\sigma z_\sigma/p_B(\sigma)$.
    Therefore $-\eps \le \ln(z_\sigma/p_B(\sigma)) \le \eps$, and hence $\e^{-\eps} \mu_\Phi(B \gets \sigma) \le z_\sigma \le \e^{\eps} \mu_\Phi(B \gets \sigma)$ for every $\sigma \in \Omega_B$.
\end{proof}

For a rational number $\theta \in (0, 1)$, it suffices to choose integers $R_\theta,P_\theta$ satisfying
\begin{align} \label{eq:coloring-depth-accuracy-parameters}
    R_\theta &\ge \left\lceil \frac{\ln(4C_{\mathrm{leaf}} n^6 \ln n/((1 - \rho)\theta))}{\ln(1/\rho)}\right\rceil, \notag\\
    P_\theta &\ge \max\;\set{3, \left\lceil\log_2 \frac{32(q - 1)n^6}{\theta(1 - \rho)}\right\rceil}.
\end{align}
These choices can be made using rational arithmetic with constants depending only on $\eta,d,q$.
Fix rational numbers $\bar\rho\in(\rho,1)$, $A>C_{\mathrm{leaf}}/(1-\rho)$ and $F>(q-1)/(1-\rho)$.
Choose $R_\theta$ as the least nonnegative integer such that $4An^7\bar\rho^{R_\theta}\le\theta$, and $P_\theta$ as the least integer at least $3$ such that $2^{P_\theta}\ge32Fn^6/\theta$.
Since $\ln n\le n$ and $\rho<\bar\rho$, these integers satisfy \eqref{eq:coloring-depth-accuracy-parameters} and are $O_{\eta,d,q}(\ln(n/\theta))$.
Now we are ready to prove \Cref{thm:random-graph-coloring}.
\begin{proof}[Proof of \Cref{thm:random-graph-coloring}]
    Assume the joint event of \Cref{lem:random-graph-colorability,lem:path-weight-concentration,lem:exceptional-component-size,lem:exceptional-incident-volume,lem:block-walk-weight}.
    The event occurs with probability at least $1 - o_n(1) - 5n^{-4} = 1 - o_n(1)$.

    Set $R = R_{\eps/(32n)}$ and $P = P_{\eps/(32n)}$ like \eqref{eq:coloring-depth-accuracy-parameters}.
    It is clear to verify that $R = O(\ln{(n/\eps)})$ and $P = O(\ln{(n/\eps)})$ where the hidden constants depend only on $\eta, d$ and $q$.
    If $R \ge n$, it must hold that $\eps = \e^{-\Omega(n)}$ and we enumerate all $q^n$ assignments and check whether it is a proper coloring.
    This is a deterministic algorithm to exactly compute the number of proper colorings in $\poly{n, \eps^{-1}}$ bit operations.

    Otherwise, using the algorithm in \Cref{lem:random-graph-colorability}, we obtain a proper $q$-coloring $\sigma$ of $G$.
    The approximation scheme $\mathsf A_{\eta, d, q}$ runs as follows:
    \begin{enumerate}
        \item initialize $\Phi = (G_\Phi = G, (L_\Phi(v) = [q])_{v \in V})$;
        \item process the components of $G[X]$ in a preset order:
        \begin{enumerate}
            \item for the current component $B$, construct the rational polytope $\polytope = \polytope_R(\Phi, B)$ and record $\widehat p_{\sigma(B)} = \min_{z \in \polytope} z_{\sigma(B)}$;
            \item pin $B$ to the partial coloring $\sigma(B)$ in $\Phi$;
        \end{enumerate}
        \item process the remaining vertices of $G$ in a preset order:
        \begin{enumerate}
            \item for the current vertex $v$, construct the rational polytope $\polytope = \polytope_R(\Phi, \{v\})$ and pick $\widehat p_{\sigma(v)} = \min_{z \in \polytope} z_{\sigma(v)}$;
            \item pin $v$ to color $\sigma(v)$ in $\Phi$;
        \end{enumerate}
        \item output the rational product of reciprocals of recorded rational numbers.
    \end{enumerate}

    We first show the accuracy of the algorithm.
    Suppose that the instances in the algorithm are $\Phi = \Phi_0, \ldots, \Phi_M = \varnothing$, the true values of the marginal probabilities are $p_0, \ldots, p_{M - 1}$ and the estimations by the algorithm are $\widehat p_0, \ldots, \widehat p_{M - 1}$.
    Note that all needed analysis holds through the running of $\mathsf A_{\eta, d, q}$.
    By \Cref{lem:coloring-accuracy-bound,rem:coloring-residual-accuracy}, it holds that $\e^{-\eps/(32n)} \le \widehat p_j/p_j \le 1$.
    Hence
    \begin{align*}
        1 \le \frac{\prod_{j = 0}^{M - 1} \widehat p_j^{-1}}{\abs{\Omega(G, q)}} = \frac{\prod_{j = 0}^{M - 1} \widehat p_j^{-1}}{\prod_{j = 0}^{M - 1} p_j^{-1}} \le \e^{\eps/16} \le 1 + \eps.
    \end{align*}

    Now we show the bit complexity.
    An LFP over a block polytope on $B$ has $O(q^{\abs B})$ variables and $O(q^{2\abs B})$ linear constraints.
    By \Cref{fact:coloring-polytope-soundness}, its encoding length is $O(q^{4\abs B}n^2P)$.
    The polynomial-time LFP bound in \Cref{thm:Charness-Cooper-transformation} and the rounding oracle therefore give the following crude bound for one child optimization:
    \begin{align*}
        q^{50\abs B} (nP)^{40}
    \end{align*}
    bit operations.
    The local work to form the constraints at a starting block $B$, together with the final coordinate minimization there, is bounded by $q^{40\abs B}(nP)^{40}$.
    We charge each optimization over a recursively constructed child polytope to that child.

    We turn to the total bit complexity to estimate $p_i$.
    A recursive branch is encoded by the starting block $B$, its first boundary edge, a permissive block walk of length at most $R$, and the pairs of partial colorings selected on the blocks.
    There are at most $q^{2\abs{B_j}}$ choices for the pair on $B_j$.
    Consequently the total cost is bounded by
    \begin{align}\label{eq:coloring-block-cost}
        O\ab(q^{40\abs{B}} \ab(1+\sum_{v \in B} \deg_G(v)) (nP)^{80}
        \ab(1+\sum_{\substack{\text{permissive block walks}\\ B_0,\ldots,B_\ell:\,\ell\le R}}
        \prod_{j=0}^{\ell}q^{52\abs{B_j}})).
    \end{align}
    The added constants include the local work when $B$ has no boundary.
    In the outer algorithm, $B$ is either a component of $G[X]$ or a singleton outside $X$ after $X$ has been pinned.
    On the structural event, $\abs B\le4\ln n$ and $\sum_{v\in B}\deg_G(v)\le C_{\mathrm{vol}}\ln n+q-2$.
    Since $K=200\ge52$, \Cref{lem:block-walk-weight} bounds the resulting total weight in \eqref{eq:coloring-block-cost}, uniformly over all accessible starting instances and therefore simultaneously for all outer instances.
    In particular, $q^{40\abs B}\le n^{160\ln q}$, so the bit complexity to estimate a marginal is at most
    \begin{align*}
        O\ab(n^{160\ln q}(1+C_{\mathrm{vol}}\ln n+q)(nP)^{80}
        \ab(1+\sum_{\ell=0}^R\mathcal W_\ell(G)))
        \le \ab(\frac{n+\bit(\eps)}{\eps})^C
    \end{align*}
    for a constant $C$ depending only on $\eta,d,q$, by the choices $R,P=O_{\eta,d,q}(\ln(n/\eps))$.
    There are at most $n$ outer queries, and the bit lengths of the recorded rational optima and their product are polynomially bounded by the same estimates.
    Increasing $C$ absorbs these operations and the input precision, proving the desired bit complexity of $\mathsf A_{\eta,d,q}$.
\end{proof}

\section*{Acknowledgement}
We thank Guoliang Qiu and Chihao Zhang for helpful discussion on counting proper colorings of random graphs.
We also thank Lutz Warnke for pointing out appropriate references for Lagrange inversion.

\bibliographystyle{alpha}
\bibliography{refs}

\appendix
\section{Proof of Contractions} \label{sec:contraction-proof}

\subsection{Birkhoff contraction of Hilbert distance}
Now we show \Cref{lem:Birknoff-contraction}.
By \Cref{lem:TV-distance-by-Hilbert-distance}, for two distinct probability vectors $p, q$,
\begin{align*}
    \TV{p}{q} \le \tanh(\delta/4)
\end{align*}
where $\delta = d_H(p, q)$.
Define $r_k = p(k) - q(k)$.
For any real vector $h$, set $m = \min_k h_k$ and $M = \max_k h_k$.
It holds that $\sum_{r_k > 0} r_k = -\sum_{r_k < 0} r_k = \TV{p}{q}$.
Therefore,
\begin{align} \label{eq:inner-bound}
    \sum_k r_k h_k \le M \sum_{r_k > 0} r_k + m \sum_{r_k < 0} r_k = (M - m) \TV{p}{q}.
\end{align}
Since \eqref{eq:inner-bound} also holds for $-h$, it holds that
\begin{align} \label{eq:inner-absolute-bound}
    \abs{\sum_k (p(k) - q(k)) h_k} \le \tanh{(\delta/4)} \ab(\max_k h_k - \min_k h_k).
\end{align}
Set $h_k = \ln(x(k)/y(k))$ and $z(t)_k = y(k) \e^{t h_k}$.
For every $a$, define
$$
    F_a(t) = \ln (Az(t))_a, \quad p^{a, t}(k) = \frac{A(a, k) z(t)_k}{(Az(t))_a}.
$$
It holds that $p^{a, t}$ is a positive probability vector and $F_a'(t) = \sum_k p^{a, t}(k) h_k$.
Note that for $a, b$,
$$
    \frac{\max_i p^{a, t}(i)/p^{b, t}(i)}{\min_j p^{a, t}(j)/p^{b, t}(j)} = \max_{i, j} \frac{A(a, i) A(b, j)}{A(a, j) A(b, i)} \le \e^{\Delta_B(A)}.
$$
Then applying \eqref{eq:inner-absolute-bound} to $p^{a, t}$ and $p^{b, t}$ and integrating from $0$ to $1$, it holds that
\begin{align*}
    \abs{(F_a(1) - F_a(0)) - (F_b(1) - F_b(0))} &\le \int_0^1 \abs{F_a'(t) - F_b'(t)} \d{t} \\
    &\le \tanh(\Delta_B(A)/4) \ab(\max_k h_k - \min_k h_k) \\
    &= c_B(A) d_H(x, y).
\end{align*}
Observe that $F_a(1) - F_a(0) = \ln\frac{(Ax)_a}{(Ay)_a}$.
Thus we conclude the lemma.

\subsection{Contraction of Hilbert distance for list colorings}
In this part, we prove \Cref{lem:coloring-Hilbert-contraction}.
When $p = p'$, the upper bound is trivial.

Assume that $p \neq p'$.
Let $p(t) = (1 - t)p' + tp$ for $t \in [0, 1]$.
It holds that $p(t)$ is a positive probability vector on the active support $\Omega_+$ and for every $c \in \Omega_+$, $p_c(t) \le 1/s$.
For $c\in\Omega_+$, let $d_c = p_c - p'_c$ and $x_c(t) = d_c / p_c(t)$.
For $c\in\Omega\setminus\Omega_+$, set $p_c(t)=d_c=x_c(t)=0$; these coordinates remain identically zero and are excluded from all Hilbert distances, minima, and maxima below.
Then $\sum_{c\in\Omega_+} p_c(t) x_c(t) = 0$, so $\min_{c\in\Omega_+} x_c(t) \le 0 \le \max_{c\in\Omega_+} x_c(t)$.
Let
\begin{align*}
    D_{a, b}(t) \defeq \frac{\d}{\d{t}}\ab(\ln\frac{1 - p_a(t)}{1 - p_b(t)}) = -\frac{p_a(t)}{1 - p_a(t)} x_a(t) + \frac{p_b(t)}{1 - p_b(t)} x_b(t).
\end{align*}
Put $m(t) = \min_{c\in\Omega_+} x_c(t)$ and $M(t) = \max_{c\in\Omega_+} x_c(t)$.
Moreover, set
\begin{align*}
    \alpha = \frac{p_a(t)}{1 - p_a(t)}, \quad \beta = \frac{p_b(t)}{1 - p_b(t)}.
\end{align*}
It is clear that $\alpha, \beta \in [0, 1/(s - 1)]$ and therefore $\alpha x_a(t), \beta x_b(t) \in [m(t)/(s - 1), M(t)/(s - 1)]$ since $m(t) \le 0 \le M(t)$.
If $a$ or $b$ lies outside $\Omega_+$, its corresponding term is zero, which belongs to the same interval; thus the derivative formula and the ensuing bound still apply.
Then for every $a, b \in \Omega$ and $t \in [0, 1]$,
\begin{align*}
    \abs{D_{a, b}(t)} \le \frac{M(t) - m(t)}{s - 1}.
\end{align*}
Write $R_c \defeq p_c / p'_c$ for every $c \in \Omega_+$, and take the extrema of $R_c$ below over this support.
It holds that
\begin{align*}
    x_c(t) = \frac{d_c}{p_c(t)} = \frac{p_c - p'_c}{(1 - t)p'_c + t p_c} = \frac{R_c - 1}{t R_c + (1 - t)}.
\end{align*}
With the observation that $x_c(t)$ is increasing with $R_c$, we have
\begin{align*}
    \abs{\ln\frac{1 - p_a}{1 - p'_a} - \ln\frac{1 - p_b}{1 - p'_b}} &\le \frac{1}{s - 1} \int_0^1 (M(t) - m(t)) \d{t} \\
    &= \frac{1}{s - 1} \int_0^1 \ab(\frac{\max_c R_c - 1}{t \max_c R_c + (1 - t)} - \frac{\min_c R_c - 1}{t \min_c R_c + (1 - t)}) \d{t} \\
    &= \frac{1}{s - 1} \ab(\ln \max_c R_c - \ln \min_c R_c) \\
    &= \frac{d_H(p, p')}{s - 1},
\end{align*}
where the last equality comes directly from the definition of $d_H(\cdot, \cdot)$.

\section{Parameter Choices for Random Graph Colorings}
\label{sec:coloring-parameter-choice}

We verify that all parameter estimates used in
\Cref{sec:random-graph-coloring} hold for a common threshold
$d_0(\eta)$, uniformly over $q \ge (2+\eta)d$.

\begin{lemma}[Parameter choice]\label{lem:coloring-parameter-choice}
For every fixed $\eta\in(0,1)$, there exists $d_0=d_0(\eta)\ge8$
such that, for every $d\ge d_0$ and every integer
$q\ge(2+\eta)d$, the parameters defined in
\eqref{eq:coloring-degree-parameters},
\eqref{eq:coloring-tail-parameters}, and
\eqref{eq:coloring-block-parameters} satisfy
\eqref{eq:coloring-threshold-estimates},
\eqref{eq:coloring-path-parameter-estimate},
\eqref{eq:coloring-component-parameter-estimate}, and
\eqref{eq:coloring-block-parameter-estimate}.
\end{lemma}

\begin{proof}
Recall that $K=200$ is an absolute constant.
Set $j_\eta=(2+\eta)\ln2-1>0$.
Choose $d_0\ge8$ so that, for every $d\ge d_0$,
\begin{equation}\label{eq:coloring-d-setting}
\begin{aligned}
    \eta d &\ge24,\\
    d\e^{-dh(\eta/6)/2} &\le \frac{\eta}{2(3+\eta)},\\
    4\e d\e^{-dh(\eta/6)/4} &\le1,\\
    d &\ge \frac{2K}{(2+\eta)\ln2},\\
    j_\eta d-\ln2
    &\ge 2\ln(16\e)+2(K+1)\ln d+2K\ln(2+\eta).
\end{aligned}
\end{equation}
Such a choice exists because $h(\eta/6)>0$ and $j_\eta>0$:
the exponential terms dominate the polynomial factors,
and the linear term in the last inequality dominates $\ln d$.
All these conditions depend only on $\eta$ and the fixed constant $K$.

First, $\delta\ge\eta/6$ and $\delta d\ge4$.
Since $(d+q)/3=(1+2\delta)d$, we have
\[
    D-3>(1+2\delta)d-4\ge(1+\delta)d.
\]
Moreover,
\[
    q-D-1
    \ge \frac{2q-d}{3}-1
    \ge \frac{(3+2\eta)d}{3}-1
    \ge \frac{(3+\eta)d}{3}.
\]
In particular, $D<q-1$ and
\[
    d\gamma\le\frac{3}{3+\eta}.
\]
This proves \eqref{eq:coloring-threshold-estimates}.

Next, $h'(x)=\ln(1+x)>0$ for $x>0$, so
$r\le\e^{-dh(\eta/6)}$.
Consequently,
\[
    d(\gamma+\sqrt r)
    \le \frac{3}{3+\eta}+\frac{\eta}{2(3+\eta)}
    =\lambda_0<\rho<1,
\]
and
\[
    \e d r^{1/4}
    \le \e d\e^{-dh(\eta/6)/4}
    \le\frac14.
\]
These are \eqref{eq:coloring-path-parameter-estimate}
and \eqref{eq:coloring-component-parameter-estimate}.

Finally, put $q_0=(2+\eta)d$.
By \eqref{eq:coloring-d-setting}, $q_0\ge2K/\ln2$.
The function $u\mapsto u^K\e^{-u\ln2/2}$ is nonincreasing
for $u\ge q_0$, since its logarithmic derivative is
$K/u-\ln2/2\le0$.
It follows that
\begin{align*}
    \e y
    &=4\e d\,\e^{(d+\ln2)/2}q^K\e^{-q\ln2/2}\\
    &\le4\e d\,((2+\eta)d)^K
          \e^{-(j_\eta d-\ln2)/2}\\
    &\le\frac14,
\end{align*}
where the last step uses the final inequality in
\eqref{eq:coloring-d-setting}.
This proves \eqref{eq:coloring-block-parameter-estimate}
with a threshold $d_0$ independent of $q$.
\end{proof}

\end{document}